\documentclass[11pt,a4paper]{article}

\usepackage{amsfonts}
\usepackage{amsmath}
\usepackage{epsfig,amssymb,amsmath,latexsym}
\usepackage[all]{xy}
 \usepackage[round]{natbib}

\usepackage{natbib}

\def\a{{\alpha}}

\def\l{{\lambda}}
\def\s{{\sigma}}
\def\t{{\tau}}
\def\g{{\gamma}}

\newcommand{\be}{\begin{equation}}
\newcommand{\ee}{\end{equation}}
\newcommand{\bea}{\begin{eqnarray}}
\newcommand{\eea}{\end{eqnarray}}
\newcommand{\beas}{\begin{eqnarray*}}
\newcommand{\eeas}{\end{eqnarray*}}

\newtheorem{theorem}{Theorem}[section]
\newtheorem{definition}[theorem]{Definition}
\newtheorem{proposition}[theorem]{Proposition}
\newtheorem{corollary}[theorem]{Corollary}
\newtheorem{lemma}[theorem]{Lemma}
\newtheorem{remark}[theorem]{Remark}
\newtheorem{example}[theorem]{Example}
\newtheorem{examples}[theorem]{Examples}
\newtheorem{foo}[theorem]{Remarks}
\newenvironment{Example}{\begin{example}\rm}{\end{example}}

\newenvironment{Remark}{\begin{remark}\rm}{\end{remark}}

\newenvironment{proof}{\addvspace{\medskipamount}\par\noindent{\it Proof}.}
{\unskip\nobreak\hfill$\Box$\par\addvspace{\medskipamount}}

\newcommand{\E}[1]{{\mathbb{E}}\left[#1\right]}
\newcommand{\Ehatf}[1]{{\hat{\mathbb{E}}_{\F_0}}\left[#1\right]}
\newcommand{\Ebarf}[1]{{\bar{\mathbb{E}}_{\F_0}}\left[#1\right]}

\newcommand{\EbarStau}[1]{{\mathbb{E}}_{\bar{\F}^S_\tau}\left[#1\right]}
\newcommand{\EFn}[1]{{\mathbb{E}_{\F^n}}\left[#1\right]}
\newcommand{\Et}[1]{{\mathbb{E}_{\F_t}}\left[#1\right]}
\newcommand{\EFS}[1]{{\mathbb{E}_{\F^S}}\left[#1\right]}

\newcommand{\EbarQSs}[1]{{\mathbb{E}^{\mathbb{Q}}_{\bar{\F}^S_s}}\left[#1\right]}

\newcommand{\Ejh}[1]{{\mathbb{E}}_{\F_{jh}}\left[#1\right]}
\newcommand{\Ei}[1]{{\mathbb{E}}_{\F_{ih}}\left[#1\right]}

\newcommand{\EQti}[1]{{\mathbb{E}}^{\mathbb{Q}^h}_{\F_{ih}}\left[#1\right]}
\newcommand{\EQsigma}[1]{{\mathbb{E}}^{\mathbb{Q}}_{\F_{\s}}\left[#1\right]}
\newcommand{\EQs}[1]{{\mathbb{E}}^{\mathbb{Q}}_{\F_{s}}\left[#1\right]}
\newcommand{\EQt}[1]{{\mathbb{E}}^{\mathbb{Q}}_{\F_{t}}\left[#1\right]}
\newcommand{\ESip}[1]{{\mathbb{E}}^{\mathbb{Q}^h}_{\F^{S^h}_{(i+1)h}}\left[#1\right]}
\newcommand{\Enip}[1]{{\mathbb{E}}_{\F^{S^h}_{(i+1)h}}\left[#1\right]}
\newcommand{\ESjp}[1]{{\mathbb{E}}_{\F^{S^h}_{(j+1)h}}\left[#1\right]}
\newcommand{\EG}[1]{{\mathbb{E}}_{\G}\left[#1\right]}
\newcommand{\EbarG}[1]{{\mathbb{E}}_{\G}\left[#1\right]}

\newcommand{\Ebarr}[1]{{\mathbb{\bar{E}}}_{\G}\left[#1\right]}
\newcommand{\EQG}[1]{{\mathbb{E}}^{\mathbb{Q}_{\G}}\left[#1\right]}
\newcommand{\EQh}[1]{{\mathbb{E}}^{\mathbb{Q}^h}_{\F_{ih}}\left[#1\right]}
\newcommand{\EQj}[1]{{\mathbb{E}}^{\mathbb{Q}^h}_{\F_{jh}}\left[#1\right]}

\newcommand{\EQz}[1]{{\mathbb{E}}^{\mathbb{Q}}\left[#1\right]}
\newcommand{\EGS}[1]{{\mathbb{E}}_{\F^S}\left[#1\right]}
\newcommand{\Ebar}[1]{{\bar{\mathbb{E}}}\left[#1\right]}

\newcommand{\EtildeQG}[1]{{\mathbb{E}}^{\tilde{\mathbb{Q}}_{\G}}\left[#1\right]}
\newcommand{\EQst}[1]{{\mathbb{E}}^{\bar{\mathbb{Q}}_{\G}}\left[#1\right]}

\DeclareMathOperator{\esssup}{ess\,sup}
\DeclareMathOperator{\essinf}{ess\,inf}

\def\F{\mathcal{F}}
\def\G{\mathcal{G}}

\begin{document}
\title{Time-Consistent and Market-Consistent Evaluations}

\author{\textbf{Mitja Stadje}\\
Tilburg University, CentER \& Netspar\\
Dept.~of Econometrics and Operations Research\\
P.O. Box 90153\\
5000 LE Tilburg\\
The Netherlands\\
Email:~m.a.stadje@uvt.nl\\
\\
\textbf{Antoon Pelsser}\thanks{The authors would like to thank
Damir Filipovic, the participants of the AFMath2010 conference,
the DGVFM 2010 Scientific Day, and seminar participants at the
Institut Henri Poincar\'e for useful comments and suggestions
on an earlier version of this paper. We also want to thank the
referee for constructive suggestions
which led to several
improvements of the paper.}\\
Maastricht University \& Netspar\\
Dept.~of Quantitative Economics and Dept.~of Finance\\
P.O.~Box 616\\
6200 MD Maastricht\\
The Netherlands\\
Email:~a.pelsser@maastrichtuniversity.nl}

\date{%
First version: April 15, 2010\\
This version: \today}

\maketitle

\begin{abstract}
We consider evaluation methods for payoffs with an inherent
financial risk as encountered for instance for portfolios held
by pension funds and insurance companies. Pricing such payoffs
in a way  consistent to market prices typically involves
combining actuarial techniques with methods from mathematical
finance. We propose to extend standard actuarial principles by
a new market-consistent evaluation procedure which we call `two
step market evaluation.' This procedure preserves the structure
of standard evaluation techniques and has many other appealing
properties. We give a complete axiomatic characterization for
two step market evaluations. We show further that in a dynamic
setting with continuous stock prices every evaluation which is
time-consistent and market-consistent is a two step market
evaluation. We also give characterization results and examples
in terms of $g$-expectations in a Brownian-Poisson setting.

Keywords: Actuarial valuation principles, financial risk,
market-consistency, time-consistency.
\end{abstract}


\newpage

\section{Introduction}
We investigate evaluation methods for payoffs with an inherent
financial risk and propose a new market-consistent procedure to
evaluate these payoffs. Our procedure yields the extension of
many standard actuarial principles into both time-consistent
and market-consistent directions.

Many insurance companies sell products which depend on
financial as well as non-financial risk. Typical examples are
equity-linked insurance contracts or catastrophe insurance
futures and bonds. Pricing such payoffs in a way consistent to
market prices usually involves combining actuarial techniques
with methods from mathematical finance. The minimal conditions
which any market-consistent evaluation should satisfy is that a
purely financial replicable payoff should be equal to the
amount necessary to replicate it.

Standard actuarial premium principles are typically based on a
pooling argument which justifies applying the law of large
numbers to price using the expectation with respect to the
physical measure and then to take an additional risk load. With
these principles one usually considers a static premium
calculation problem: what is the price today of an insurance
contract with payoff at time $T$? See for example the textbooks
by \citet{buhlmann1970mathematical},
\citet{gerber1979introduction}, or \citet{Kaas:MART}. Also, the
concept of convex risk measures and the closely related one of
monetary utility functions have been studied in such a static
setting. See for example \citet{FoSchie:2002:Convexmeasuresof},
\citet{FrittRosaz:2002:Puttingorderin},
\citet{jouini2008optimal}, and
\citet{filipovic2008equilibrium}.

In financial pricing one usually considers a ``dynamic''
pricing problem: how does the price evolve over time until the
final payoff date $T$? This dynamic perspective is driven by
the focus on hedging and replication. This literature started
with the seminal paper of \citet{bs:optprice} and has been
immensely generalized to broad classes of securities and
stochastic processes; see \citet{Delbaen:Schachermayer:94}.

In recent years, researchers have begun to investigate risk
measures in a dynamic setting, the central question being the
construction of time-consistent (``dynamic'') risk measures.
See \citet{riedel2004dynamic:coherent},
\citet{roorda2005coherent},
\citet{cher:delb:kupp:2006:dynamic}, \citet{gianin2006risk},
\citet{artzner2007coherent}. In a dynamical context
\emph{time-consistency} is a natural approach to glue together
static evaluations. It means that the same value is assigned to
a financial position regardless of whether it is calculated
over two time periods at once or in two steps backwards in
time.  In a recent paper \citet{jobert:rogers:2008:valuations}
show how time-consistent valuations can be constructed via
backward induction of static one-period risk-measures (or
``valuations''). See also
\citet{HardyWirch:2005:IteratedCTE--aDynamic} for an example
with the risk measure given by Average Value at Risk.

An important branch of literature considers risk
measures/valuations in a so-called market-consistent setting.
This started with the pricing of contracts in an
incomplete-market setting, where one seeks to extend the
arbitrage-free pricing operators (which are only defined in a
complete-market setting) to the larger space of (partially)
unhedgeable contracts. One approach to evaluate the payoff in
this situation is by utility indifference pricing: the investor
pays the amount such that he is no worse off in expected
utility terms than he would have been without the claim. The
paper by \citet{HodgesNeuberger1989optimal} is often cited for
the root-idea of this stream of literature. For other
contributions in this direction, see for instance
\citet{Henderson:02}, \citet{Young:02},
\citet{hobson2004q-optimal}, \citet{Musiela:04},
\citet{monoyios2006characterisation}, and the recent book by
\citet{carmona2009indifference}.

Several papers deal with the extension of the arbitrage-free
pricing operators using (local) risk-minimisation techniques
and the related notion of minimal martingale measures; see
\citet{follmer1989hedging}, \citet{schweizer1995minimal},
\citet{delbaen1996variance}. A rich duality theory has been
developed that establishes deep connections between utility
maximisation and minimisation over martingale measures; see
\citet{cvitanic1992convex}, \citet{kramkov1999asymptotic}. A
 very elegant summary is given by \citet{rogers2001duality}.
Another stream of the literature, where the class of martingale
measures considered is restricted, is given by the works on
good-deal bound pricing, see \citet{CochSaaR2000GDBincom},
\citet{CernHodg2002GDB}, and \citet{BjorkSlink2006TheoGDB}.

Using utility-indifference (and duality) methods, the
market-consistency of pricing operators is automatically
induced. However, an explicit formal definition of
market-consistent pricing operators has only begun to emerge
recently; see \citet{kupper2008DynaRiskMeas},
\citet{malamud2008market},
\citet{barrieu:elkaroui2005infconvolution,
barrieu:elkaroui:2009:carmona}, and \citet{knispel2011black}.

In this paper we investigate well-known actuarial premium
principles such as the variance principle and the
standard-deviation principle, and study their extension into
both time-consistent and market-consistent directions. To do
this, we introduce the concept of \emph{two step market
evaluations} and study their properties. Two step market
evaluations convert any evaluation principle into a
market-consistent one by applying the actuarial principle to
the residual risk which remains after having conditioned on the
future development of the stock price. This operator splitting
preserves the structure and the computationally tractability of
the original actuarial evaluation. Furthermore, we get some
appealing properties like numeraire invariance. We are able to
give a complete axiomatic characterization for two step market
evaluations and show that these axioms are satisfied in a
setting where the stock process is continuous and the insurance
process is revealed at fixed time instances (or more generally
has predictable jumps). This provides a strong argument for the
use of two step market evaluations. We also consider some
time-consistent extensions of our market-consistent evaluations
to continuous time in a Brownian-Poisson setting. For this we
 need some results from the theory of backward stochastic
differential equations (BSDEs), also called $g$-expectations.
For background material on BSDEs we refer to
\citet{KarouiPengQuenez1997BSDE-Finance}.

The paper is organized as follows. In Section 2 we define
conditional evaluations, give some background material, and
recall some of the most standard actuarial principles. In
Section 3 the notion of market-consistency is defined and two
step market evaluations are introduced and motivated. We give a
complete axiomatic characterization for two step market
evaluations. In Section 4 it is shown that in a dynamic setting
with continuous stock prices every evaluation which is
time-consistent and market-consistent can be viewed as a two
step market evaluation. In Section 5 we extend our evaluations
to a continuous-time setup with processes with jumps. Section 6
gives a summary and conclusions. Section 7 contains the proofs
of our results.
\section{Conditional Evaluations}
Let $(\Omega,\F,\mathbb{P})$ be a probability space. Equalities
and inequalities between random variables are understood in the
$\mathbb{P}$-almost sure sense unless explicitly stated
otherwise. The space of bounded random variables will be
denoted by $L^\infty(\Omega, \F, \mathbb{P})$ ($ L^\infty(\F)$
for short). The space of bounded, non-negative random variables
will be denoted by $L^\infty_+(\F).$ The space of random
variables which are integrable with respect to $\mathbb{P}$
will be denoted by $L^1(\Omega, \F, \mathbb{P})$ ($ L^1(\F)$
for short). \emph{Financial and insurance positions} are
represented by random variables $H\in L^\infty(\F)$ where
$H(\omega)$ is the discounted net loss of the position at
maturity under the scenario $\omega.$
Now given a $\s$-algebra $\G\subset\F,$ with information
available to the agent, we can define a \emph{conditional
evaluation}:
\begin{definition}
A mapping $\Pi_{\G}:L^\infty(\F) \to L^\infty(\G)$ is called a
$\G$-conditional evaluation if the following axioms hold:
\begin{itemize}
\item {\it Normalization}: $\Pi_{\G}(0)=0$.
\item $\G${\it -Cash Invariance}:
    $\Pi_{\G}(H+m)=\Pi_{\G}(H)+m$ for $H\in L^\infty(\F)$
    and $m\in  L^\infty(\G)$.
\item $\G${\it -Convexity}: For $H_1,H_2\in L^\infty(\F)$
    $\Pi_{\G}(\l H_1+(1-\l)H_2)\leq \l
    \Pi_{\G}(H_1)+(1-\l)\Pi_{\G}(H_2)$ for all $\l\in
    L^\infty(\G)$ with $0\leq \l\leq 1$.
\item $\G${\it-Local Property:} $\Pi_{\G}(I_A
    H_1+I_{A^c}H_2)=I_A\Pi_{\G}(H_1)+I_{A^c}\Pi_{\G}(H_2)$
    for all $H\in L^\infty(\F)$ and $A\in \G.$
\item {\it Fatou property}: For any bounded sequence
    $(H_n)$ which converges a.s. to $H$
$$\Pi_{\G}(H)\leq \liminf_n\Pi_{\G}(H_n).$$
\end{itemize}
\end{definition}
Normalization guarantees that the null position does not
require any capital reserves. If $\Pi$ is not normal then the
agent can consider the operator $\Pi(H)-\Pi(0)$ without
changing his preferences. Convexity, which under cash
invariance is equivalent to quasiconvexity, says that
diversification should not be penalized.
Cash invariance gives the interpretation of $\Pi(H)$ as a
capital reserve. The local property is motivated in the
following way. Since the agent has the information given by
$\G$ he knows if the event $A$ has happened or not and should
adjust his evaluation accordingly. If $\Pi$ satisfies
\begin{itemize}
\item {\it Monotonicity}: For $H_1,H_2\in L^\infty(\F)$
    with $H_1\leq H_2$ $\Pi_{\G}(H_1)\leq \Pi_{\G}(H_2)$
\end{itemize}
then we will call $\Pi$ a monotone conditional evaluation.
Monotonicity postulates that if in a.s. any scenario $H_2$
causes a greater higher loss than $H_1$ then the premium
charged for $H_2$ (or the capital reserve held) should be
greater than for $H_1.$ Note that if $\Pi_\G$ is monotone then
$\rho_\G(H):=\Pi_\G(-H)$ defines a conditional convex risk
measure and $U_\G(H):=-\rho_\G(H)$ defines a conditional
monetary utility function. For the definition of a convex risk
measure, see \citet{FoSchie:2002:Convexmeasuresof}, or
\citet{FrittRosaz:2002:Puttingorderin}. In particular, all
results in this paper also hold (with obvious change of signs)
for conditional convex risk measures and conditional monetary
utility functions.

\begin{Remark}
It has been shown in \citet{cher:delb:kupp:2006:dynamic} that
the local property must be satisfied if $\Pi_\G$ is monotone
and cash invariant. Moreover, the local property is also
implied by convexity. Indeed, if $\Pi_\G$ is convex, then we
have for $A\in\G$ that clearly $\Pi_\G(1_A H_1 + 1_{A^c} H_2)
\le 1_A \Pi_\G (H_1) + 1_{A^c} \Pi_\G(H_2).$ In particular, $
 1_A \Pi_\G(1_A
H_1 + 1_{A^c} H_2) \le 1_A \Pi_\G(H_1) .$ The other direction
follows by setting $\tilde{H} = 1_A H_1 + 1_{A^c} H_2$ then as
before
$$
 1_A \Pi_\G(H_1)=1_A \Pi_\G(1_A
\tilde{H} + 1_{A^c} H_2) \le 1_A \Pi_\G(\tilde{H}) .$$ Switching
the role of $H_1$ and $H_2$ yields then the desired conclusion.
\end{Remark}

Other possible axioms which we will consider in a dynamic
setting are as follows:
\begin{itemize}
\item {\it $\G$-Positive Homogeneity}: For $H\in
    L^\infty(\F)$ $\Pi_{\G}(\l H)= \l \Pi_{\G}(H)$ for all
    $\l\in L^\infty_+(\G).$
\item {\it Continuity}: For any bounded sequence $(H_n)$
    which converges a.s. to $H$
$$\Pi_{\G}(H)= \lim_n\Pi_{\G}(H_n).$$
\item {\it $p$-norm boundedness}: There exists
    $p\in(1,\infty),$ $\l\in L^\infty_+(\G)$, and a measure
    $\bar{\mathbb{P}}$ having the same zero sets as
    $\mathbb{P}$ such that for $H\in L^\infty(\F)$
$$\Pi_{\G}(H)\leq \l \int (|H|+|H|^p) d\bar{\mathbb{P}}_{\G}.$$
\end{itemize}
We will also refer to the continuity axiom as `continuity with
respect to a.s. bounded convergence,' if there is any
ambiguity.
If $\Pi$ is a $\G$-conditional evaluation which is additionally
assumed to be positively homogeneous then we call $\Pi$ a
$\G$-conditional \emph{coherent} evaluation. For a further
discussion of these axioms see also \citet{ADEH:1999coherent}.
Note that many similar axioms for premium principles can be
found in the literature, see for instance
\citet{DepreGerbe:1985:convexprinciplesof} or
\citet{Kaas:MART}. Conditional evaluations in a dynamic setting
have been considered for instance in
\citet{frittelli2004dynamic}, \citet{roorda2005coherent},
\citet{ruszczynski2006conditional},
\citet{Delb2006mStableSets}, \citet{artzner2007coherent},
\citet{KlSchwe:2007:Dynamicindifferencevaluation},
\citet{jobert:rogers:2008:valuations},
\citet{barrieu:elkaroui:2009:carmona}, and
\citet{cheridito2006composition}. Classical examples of
(conditional) evaluations are, see for instance
\citet{Kaas:MART}:
\begin{examples}
\label{exdc}
\begin{itemize}
\item {\it Conditional Mean-Variance principle:}
$$\Pi^v_{\G}(H)=\EbarG{H}+\frac{1}{2}\alpha Var_{\G}[H], \ \ \a\geq 0. $$
\item {\it Conditional Standard-Deviation principle:}
$$\Pi^{st}_{\G}(H)=\EbarG{H}+\beta\sqrt{ Var_{\G}[H]}, \ \ \beta\geq 0.$$
\item {\it Conditional Semi-Deviation principle:}
$$\Pi^s_{\G}(X)=\EbarG{H}+\l \bigg|\EbarG{(H-\EbarG{H})^q_+}\bigg|^{1/q}, \ \ \l\geq 0,\,\,\,\,
\ \  q\in [1,\infty),$$ where $x_+$ is $0$ if $x<0$ and $x$
else.
\item {Conditional Average Value at Risk principle:}
$$\Pi^{AV@R}_{\G}(H)=\EbarG{H}+\delta AV@R^\alpha_{\G}(H-\EbarG{H}), \ \ \delta\geq 0$$
where $AV@R^\alpha_{\G}(H)=\dfrac{1}{\alpha}\int_0^\alpha
V@R^\lambda_{\G}(H)d\lambda, \ \ \alpha \in (0,1]$ and
$V@R^\lambda_{\G}(H)$ corresponds to computing the Value at
Risk of $H$ at the confidence level $\lambda$ with the
available information $\G.$
\item {\it Conditional Exponential principle:}
$$\Pi^v_{\G}(H)=\gamma \log\Big(\EbarG{\exp\{H/\gamma\}}\Big), \ \ \gamma> 0. $$
\end{itemize}
\end{examples}
Apart from the exponential principle the evaluations above are
generally not monotone. However, they are continuous,
$\G$-convex, $\G$-cash invariant and satisfy the local
property. In particular, they are $\G$-conditional evaluations.
The Standard-Deviation principle, the Average Value at Risk
principle and the Semi-Deviation principle additionally satisfy
$\G$-positive homogeneity while the Mean-Variance principle is
$p$-norm bounded (with $\bar{\mathbb{P}}=\mathbb{P}$). The
Average Value at Risk principle and the Semi-Deviation
principle are monotone if $\l$ or $\delta$ are in $[0,1].$

We will need some duality results. For a $\s$-algebra
$\G\subset \F,$ denote $\mathcal{Q}_{\G}=\{\xi\in
L^1(\F)|\EbarG{\xi}=1\}$ and $\mathcal{Q}^+_{\G}:=\{\xi\in
L^1_+(\F)|\EbarG{\xi}=1\}.$ In other words, given the
information $\G$, $\mathcal{Q}_{\G}$ is the set of all signed
measures and $\mathcal{Q}^+_{\G}$ is the set of all probability
measures. Therefore, conditional on our starting information
$\G$, we may identify every $\xi\in \mathcal{Q}_{\G}$ with a
signed measure, and every $\xi\in \mathcal{Q}^+_{\G}$ with a
probability measure. For instance, for $\xi\in
\mathcal{Q}^+_{\G}$ we can define the corresponding conditional
probability measure $\mathbb{Q}^\xi_{\G}(A):=\EbarG{\xi I_A},$
and its conditional density as
$\frac{d\mathbb{Q}^{\xi}_{\G}}{d\mathbb{P}_{\G}}:=\xi.$

Recall that by standard duality results we have that $\Pi_{\G}$
is coherent if and only if \be \label{coherent}
\Pi_{\G}(H)=\esssup_{\xi\in M}\EbarG{\xi H}, \ee for a unique
closed, convex set $M\subset \mathcal{Q}_{\G}.$ For the precise
definition of the essential supremum, see the Appendix. $\xi$
is often interpreted as a weighting function for the different
scenarios $\omega$, or as a test or stress measure. By taking
the supremum, a worst-case approach is being taken. For
instance in the good-deal bound literature mentioned in the
introduction the supremum is taken over all pricing kernels
with a density admitting a variance smaller than a certain
constant.

 Generally it holds for all
conditional evaluations that for $\Pi^*_{\G}$
defined by $$\Pi^*_{\G}(\xi)=\esssup_{H\in L^\infty(\F)}\{
\EbarG{\xi H}-\Pi_{\G}(H)\}$$ we have that \be \label{duality}
\Pi_{\G}(H)=\esssup_{\xi\in \mathcal{Q}_{\G}}\{ \EbarG{\xi
H}-\Pi^*_{\G}(\xi)\}, \ee see for instance
\citet{Delb2006mStableSets} or
\citet{cheridito2006composition}. $\Pi^*_{\G}$ is also called
the penalty function corresponding to $\Pi_{\G},$ and
$-\Pi^*_{\G}(\xi)$ may be seen as the plausibility of the
density $\xi$. The dual representations (\ref{coherent}) and
(\ref{duality}) are often interpreted as robust expectations
with respect to different priors, see
\citet{hansen:sargent2001robust:AER}, or
\citet{hansen:sargent2007robustness}.

Of course our results also hold without conditioning on $\G$.
In this case $\G$ is chosen to be the trivial $\s$-algebra.
However, for our dynamic analysis in the later sections it will
simplify matters if we do our analysis conditional on some
information available to the agent. \setcounter{equation}{0}
\section{Market-Consistent Pricing}
\subsection{Market-consistency and two step market evaluations}
\label{markcon} Let $(\Omega,\F,\mathbb{P})$ be the  underlying
probability space. Let $\G\subset \F$ be a $\sigma$-algebra
whose information is initially available to the agent. Let
$S=(S^1,\ldots,S^n)$ be the $n$-dimensional price process of
$n$ traded stocks and denote by $\bar{\F}^{S} \subset \F$ the
$\sigma$-algebra generated by $S$. Furthermore, we denote by
$\F^{S}$ the $\s$-algebra given by the stock process and our
starting information $\G$, i.e., $\F^{S}:=\bar{\F}^{S}\vee
\G:=\sigma(\bar{\F}^{S}, \G).$ The financial market given by
the $n$-dimensional stock process $S$ should be arbitrage free
and complete, i.e., all derivatives which conditional on $\G$
only depend on $S$ are perfectly hedgeable and there exists a
unique probability measure $\mathbb{Q}_\G\in \mathcal{Q}^+_\G$
such that $S$ is a martingale under $\mathbb{Q}_\G$
(componentwise). Furthermore, $\mathbb{Q}_\G$ is assumed to be
a.s. equivalent to $\mathbb{P}_\G$, in the sense that its
conditional density is positive.

Since $\F^S$ is in general a strict subset of $\F$, the market
given by all $\F$-measurable payoffs is incomplete. For
instance, we could have an untraded insurance process which is
correlated with the traded assets $S$ but not perfectly
replicable.


In the financial market the martingale measure $\mathbb{Q}_\G$
defines the (linear) no-arbitrage pricing operator
$\Pi^{f}_{\G}:L^\infty(\F^S)\to L^\infty(\G)$ given by
$$\Pi^{f}_{\G}(H^S):=\EQG{H^S}:=\int_{\Omega} H^S(\omega) \mathbb{Q}_{\G}(d\omega)=\EG{\xi^{ \mathbb{Q}_{\G}}H},$$
where $\xi^{ \mathbb{Q}_{\G}}$ is the density in
$\mathcal{Q}^+_{\G}$ with which $\mathbb{Q}_{\G}$ may be
identified. Note that only the martingale measure on $\F^S$ is
unique while on the filtration $\F$ there can be infinitely
many martingale measures.

The next definition extends the notion of cash invariance to
all assets traded in the financial market. For identical or
similar notions of market-consistency see also \citet{cont06},
\citet{kupper2008DynaRiskMeas}, \citet{malamud2008market}, or
\citet{artzner2010supervisory}.
\begin{definition}
An evaluation is called market-consistent if for any financial
payoff $H^S\in L^\infty(\F^S)$ and $H\in L^\infty(\F)$
$$\Pi_{\G}(H^S+H)=\EQG{H^S}+\Pi_{\G}(H).$$
\end{definition}
Market-consistency postulates that liquidly traded assets and
payoffs replicable by them should not carry any risk as they
can be converted to cash at any time. It follows immediately
from the definition that a market-consistent evaluation cannot
be `improved' by hedging.
\begin{Remark}
Our definition of market-consistency requires that we have
liquidly traded assets constituting a complete financial
market. If the financial market is not complete then there are
two possibilities which may still validate our approach:
\begin{itemize}
\item[(1.)] One could use certain financial derivatives as
    additional hedging instruments to make the financial
    market complete. For results in this direction, see
    \citet{jacod2010risk} and the references therein. Note
that in our setting $S$ has to be finite-dimensional.
Therefore, only finitely many additional hedging
    instruments are allowed. However, many stochastic
    volatility models like the Heston model can be
    completed in this way.
    \item[(2.)] One could remove certain financial assets
        as possible hedging instruments. In some cases the
        remaining assets might constitute a complete
        financial market.
\end{itemize}
\end{Remark}

The next proposition shows that market-consistency is already
implied by the assumption that \emph{purely} hedgeable
derivatives should be valued according to the amount of capital
necessary to replicate them. Furthermore, it shows that, in
case that $\Pi_{\G}$ is monotone, market-consistency is
equivalent with the no-arbitrage principle in the entire
market.
\begin{proposition}
\label{prop} For a conditional evaluation
$\Pi_{\G}:L^\infty(\F)\to L^\infty(\G)$ the following
statements are equivalent:
\begin{itemize}
\item[(i)] $\Pi_{\G}(H^S)=\EQG{H^S}$ for any financial
    payoff $H^S\in L^\infty(\F^S).$
\item[(ii)]
There exists a penalty function\footnote{A function $c$ is
called a penalty function if it is convex,
lower-semicontinuous and $\essinf c=0.$} $c_\G:
\mathcal{Q}_{\F^S}=\big \{Z\in
L^1(\F)|\EGS{Z}=1\big\}\to\mathbb{R}\cup \{\infty\}$ such
that we have for every $H\in L^\infty(\F)$ \beas
\Pi_{\G}(H)=\esssup_{Z\in \mathcal{Q}_{\F^S}}\{ \EQG{Z
H}-c_\G(Z)\}. \eeas
\item[(iii)] $\Pi_{\G}$ is market-consistent.
\end{itemize}
Furthermore, in the case that $\Pi_{\G}$ is additionally
assumed to be monotone, market-consistency is equivalent with
any payoff being evaluated between its super- and
sub-replication price.
\end{proposition}
In particular, in the case of monotonicity, we could have also
defined market-consistency by stating that the evaluation of
every payoff should respect the no-arbitrage principle.

Examples of market-consistent evaluations arise if an agent
\emph{starts} with a (usually non market-consistent) evaluation
$\Pi_\G$ and then tries to reduce his risk by hedging.
Specifically denote by $\mathcal{A}$ all admissible hedging
strategies $\pi$ (defined in an appropriate way), by
$\mathcal{M}_{\G}\subset
\mathcal{Q}^+_{\G}:=\mathcal{Q}_{\G}\cap L^1_+(\F)$ the set of
all local martingale measures, by $\tilde{S}$ the discounted
stock process, and by $\Pi^*_\G$ the penalty function
corresponding to $\Pi_\G$. It can be shown, see for instance
\citet{barrieu:elkaroui:2009:carmona}, or
\citet{toussaint2011framework}, that under appropriate
assumptions \be \label{hedging}
\bar{\Pi}_\G(H):=\essinf_{\pi\in \mathcal{A}}\Pi_\G(H+(\pi\cdot
\tilde{S})_T)=\esssup_{\bar{\mathbb{P}}_\G\in \mathcal{M}_{\G}}
\{\Ebarr{H}-\Pi^*_\G(\frac{d\bar{\mathbb{P}}_\G}{d\mathbb{P}_\G})\}.\ee
Note that by definition $\bar{\Pi}_\G$ is market-consistent.
Therefore, one way of obtaining market-consistent evaluations
is to intersect the test measures in the dual representations
(\ref{coherent}) and (\ref{duality}) above with local
martingale measures. This is a class of examples which arises
naturally when starting with an evaluation $\Pi_\G$.
However, the new evaluation $\bar{\Pi}_\G$ may be hard to
compute, and the structure and interpretation of the original
evaluation $\Pi_\G$ is lost. For instance, if payoffs are
evaluated using the mean-variance principle then
$\bar{\Pi}_\G(H)=\esssup_{\bar{\mathbb{P}}_\G\in\mathcal{M}_{\G}}\{\Ebarr{H}-\frac{1}{2\a}C_\G(\bar{\mathbb{P}}_\G|\mathbb{P}_\G)\},
$ where $C_\G$ is the relative Gini index defined by
$C_\G(\bar{\mathbb{P}}_\G|\mathbb{P}_\G)=\EG{(\frac{d\bar{\mathbb{P}}_\G}{d\mathbb{P}_\G})^2-1}.$
For an overview about mean-variance hedging, see
\citet{schweizer2010mean} and the reference therein. Now note
that two important reasons for the popularity of the
Mean-variance principle are: (a) it has a straight-forward
interpretation; (b) it is easy to compute. However, while the
new evaluation, $\bar{\Pi}_\G(H)$ is market-consistent (since
it uses risk adjusted probabilities), it is neither easy to
compute nor does it directly relate to the variance of the
payoff $H$ from which we started.

Consequently, in this paper we propose a new class of
market-consistent evaluations which we will call two step
market evaluations. Extending standard actuarial principles
with two step market evaluations will have the advantage that
the extensions can be computed easily and that the
interpretation of our starting principles can be preserved.
Furthermore, we will show that two step market evaluations in
general have many other appealing properties. A strong argument
for the use of two step market evaluations will be provided in
the Chapter 4. There we show that any insurance company which
wants to apply a market-consistent and time-consistent
evaluation, has to use a two step market evaluation, in a
setting where the stock process is continuous and the insurance
process is only revealed at fixed time instances (or more
generally at predictable stopping times).

We will start with evaluations like the one from our Examples
\ref{exdc}. Then we will give the corresponding
market-consistent evaluations which do not arise from hedging
but from operator splitting.
Namely, in a first step we compute the value of the position
$H$ by replacing the measure $\mathbb{P}|\G$ with the measure
$\mathbb{P}|\F^S$, i.e., we compute Mean-Variance principle,
the Standard-Deviation principle, etc., of $H,$ conditional on
$\G$ \emph{and} the values of the stocks $S$. Then for every
different future value of the stock price we get a different
evaluation. However, since payoff depending only on the stock
prices are perfectly hedgeable one could argue that these
remaining evaluations do not contain any risk. Therefore, the
total value of the position $H$ should be equal to the initial
capital needed to hedge the different evaluations, obtained in
the first step, which depend on $S$. This corresponds to taking
in a second step the expectation with respect to the risk
adjusted probability measure $\mathbb{Q}_\G$ coming from the
financial market. This procedure is computationally tractable
and preserves the evaluation principles considered in the
beginning. These principles are applied to the risk which
remains after having conditioned on $S.$ We then get the
following market-consistent examples:
\begin{examples}
\label{exdc2}
\begin{itemize}
\item {\it Two step Mean-Variance principle:}
$$\Pi^v_{\G}(H)=\EQG{\EGS{H}+\frac{1}{2}\alpha Var_{\F^S}[H]}, \ \ \a\geq 0. $$
\item {\it Two step Standard-Deviation principle:}
$$\Pi^s_{\G}(H)=\EQG{\EGS{H}+\beta\sqrt{ Var_{\F^S}[H]}}, \ \ \beta\geq 0.$$
\item {\it Two step Semi-Deviation principle:}
$$\Pi^s_\G(X)=\EQG{\EGS{H}+\l \bigg|\EGS{(H-\EGS{H})^q_+}\bigg|^{1/q}}, \ \ \l\geq 0,\,\,\,\, q\in [1,\infty).$$
\item {Two step Average Value at Risk principle:}
$$\Pi^{AV@R}_{\G}(H)=\EQG{\EGS{H}+\delta AV@R^\alpha_{\F^S}(H-\EGS{H})}, \ \ \delta\geq 0.$$
\item {\it Two step Exponential principle:}
$$\Pi_{\G}(H)=\EQG{\gamma \log\Big(\EGS{\exp\{H/\gamma\}}\Big)}, \ \ \gamma> 0. $$
\end{itemize}
\end{examples}
A standard deviation principle, which is different than the one
above but is also obtained by first conditioning on the future
stock price, is considered in \citet{Moller:2002:Onvaluation}.
The last example is known in the literature as the indifference
price of $H$ under an exponential utility function. It arises
in an incomplete market when an agent maximizes his exponential
utility through dynamic trading, see for instance
\citet{musiela2004valuation}. The indifference price for a
claim $H$ is then defined as the amount of cash the agent is
willing to pay for the right to receive $H$ such that he is no
worse off in expected utility terms than he would have been
without the claim. For references on indifference pricing see
the introduction. The examples above motivate the following
definition:
\begin{definition}
We call a $\Pi_{\G}:L^\infty(\F)\to L^\infty(\G)$ a two step
market evaluation if there exists an $\F^S$-conditional
valuation $\Pi_{\F^S}:L^\infty(\F)\to L^\infty(\F^S)$ such that
\be \label{rep} \Pi_{\G}(H)=\EQG{\Pi_{\F^S}(H)}. \ee
\end{definition}
Note that in case that there is no financial market, i.e.,
$S=0,$ our two step evaluations reduce to the standard
actuarial principles. On the other hand if $\F^S=\F$, i.e., if
the financial market gives the entire filtration, of course
$\Pi_\G(H)=\EQG{H}.$
\begin{Example}\label{GDB}
Another example for a two step evaluation defined above arises
when combining hedging with an Average Value at Risk principle.
Specifically set
$$\Pi_\G(H):=\essinf_{\pi\in\mathcal{A}}\Pi^{AV@R}_\G(H+(\pi\cdot\tilde{S})_T).$$
Define
$$M:=\bigg\{\bar{\mathbb{Q}}_\G\in \mathcal{M}_\G\Big|1-\delta\leq
\frac{d\bar{\mathbb{Q}}_\G}{d\mathbb{P}_\G}\leq
1+\delta\frac{1-\a}{\a}\bigg\},$$ where $\delta>0$ and the
risk level $\a\in(0,1]$ have been introduced in the Examples
\ref{exdc}. (Note that if $\delta=0$ then
$\Pi^{AV@R}_\G(H)=\EG{H}.$)

From (\ref{hedging}) and the dual representation of the
Average-Value at Risk principle (see for instance
\citet{follmer2004stochastic}) it follows that \beas
\Pi_\G(H)=\esssup_{\bar{\mathbb{Q}}_\G\in M}\EQst{H}.\eeas Let
$$M':=\bigg\{Z \in \mathcal{Q}^+_{\F^S} \Big| ( 1-\delta)
 \frac{d\mathbb{P}_\G}{d\mathbb{Q}_\G}
\leq Z \leq \frac{d\mathbb{P}_\G}{d\mathbb{Q}_\G}(
1+\delta\frac{1-\a}{\a})\bigg\}.$$ Note that $
\frac{d\mathbb{Q}_\G}{d\mathbb{P}_\G}M'=M$. Furthermore, $M'$
satisfies the concatenation property on $\F^S$, also called
rectangular property or $m$-stability\footnote{The rectangular
property or $m$-stability on a $\sigma$-algebra $\bar{\G}$
postulates that for every $A\in \F^S$ and $Z_1,Z_2\in M'$ we
have that $I_A Z_1+I_{A^c} Z_2\in M',$ see for instance
\citet{ChenEpste:2002:Ambiguityriskand} or
\citet{Delb2006mStableSets}. }. It may be seen that as a result
\beas  \Pi_{\F^S}(H):=\esssup_{Z\in M'}\EGS{ZH}\eeas is an
$\F^S$-conditional evaluation and it holds that
$$\Pi_\G(H)=\EQG{\Pi_{\F^S}(H)},$$ compare also with Theorem
\ref{main} below. In particular, combining hedging with the
Average-Value at Risk principle gives an example of the two
step procedure we explained above.
\end{Example}
\begin{Example}
Our last example is given by the super-replication price of a
contingent claim. The super-replication price is given by
$$\Pi_\G(H):=\esssup_{\bar{\mathbb{Q}_\G}\in \mathcal{M}_\G}\EQst{H},$$
It is straight-forward to check that
$\mathcal{M}_\G=\frac{d\mathbb{Q}_\G}{d\mathbb{P}_\G}\mathcal{Q}^+_{\F^S}.$
Clearly, \beas \Pi_{\F^S}(H):=\esssup_{Z\in
\mathcal{Q}^+_{\F^S}}\EGS{ZH}\eeas is an $\F^S$-conditional
evaluation. It computes the essential supremum conditional on
$\F^S.$ Since $\mathcal{Q}^+_{\F^S}$ is $m$-stable it can be
seen that
$$\Pi_\G(H)=\EQG{\Pi_{\F^S}(H)}.$$
In particular, the super-replication price is a two step
evaluation.
\end{Example}
\begin{Remark}
Note that equity linked insurance payoffs are typically of the
form $H=Y_T^{(n)} f(S_T)$ where $f(S_T)$ is a financial
derivative and $Y_T^{(n)}$ are the number of policy holder who
survived up to time $T$ (out of an initial cohort of $n$), see
for instance \citet{Moller:2002:Onvaluation}. In the special
case that the financial and the mortality risk are independent
a two step market evaluation would yield
$$\Pi_\G(H)=\EQG{\Pi_{\F^S}(f(S_T)Y_T^{(n)} )}=\EQG{f(S_T)}\Pi_\G(Y^{(n)}_T) .$$
Note that the structure obtained in this special case is
similar to the market-consistent valuation method suggested in
Chapter 2.6 in \citet{bookmarketcon}.
\end{Remark}

Two step market valuations provide a rich class of
market-consistent evaluations with a clear underlying
intuition. They appear as indifference-price with an
exponential utility, in hedging with an Average Value at Risk
principle, and in the super-replication price of an contingent
claim. Two step market valuations are also useful in
optimization problems since the maximum may be computed by a
two step procedure using Bellman's principle: first one can
compute a value function conditioned on the stock process, and
then in a second step one can compute the optimum by maximizing
the value function obtained in Step 1 under the pricing
measure.

Note that contrary to other evaluations two step market
evaluations can be directly converted into an equivalent
evaluation which takes the stock process as the numeraire. That
is, assume payoffs $\tilde{H}$ are expressed in units of the
$i$-th stock, $S^i$ for $i\in\{1,\ldots,n\}$, i.e.,
$\tilde{H}=H/S^i_T$. An agent who wants to use the evaluation
$\Pi_\G$ but wants to express everything in units of $S^i$
obviously should use the evaluation $\tilde{\Pi}_\G$ which
satisfies $S^i_0\tilde{\Pi}_\G(\tilde{H})=\Pi_\G(H).$ The
evaluation $\tilde{\Pi}_\G$ often might not be easy to
characterize directly, for instance with a dual representation
and a penalty function. However, a two step market evaluation
remains a two step market evaluation under the change of
numeraire. Moreover, one can just define the new penalty
function of $\tilde{\Pi}_{\F^S}$ as the penalty function of
$\Pi_{\F^S}$ in units of $S^i_T$. That is we set
$$\tilde{\Pi}^*_{\F^S}(\xi):= \frac{\Pi^*_{\F^S}(\xi)}{S^i_T}\mbox{ and }\tilde{\Pi}_{\F^S}(\tilde{H}):=\esssup_{\xi\in L^1(\F)}
\Big\{\EFS{\tilde{H}}-\tilde{\Pi}^*_{\F^S}(\xi)\Big\}.$$ Denote
by $\tilde{\mathbb{Q}}_{\G}$ the unique equivalent martingale
measure on $\F^S$ with numeraire $S^i$ that is
$$\frac{d\tilde{\mathbb{Q}}_{\G}}{d\mathbb{P}_{\G}}:=\frac{S^i_T}{S^i_0}\frac{d\mathbb{Q}_{\G}}{d\mathbb{P}_{\G}}.$$
Then we obtain
$$\tilde{\Pi}_\G(\tilde{H})=\frac{\Pi_\G(H)}{S^i_0}=\EQG{\frac{S^i_T}{S^i_0}\esssup_{\xi\in L^1(\F)}
\Big\{\EFS{\tilde{H}}-\tilde{\Pi}^*_{\F^S}(\xi)\Big\}}=\EtildeQG{\tilde{\Pi}_{\F^S}(\tilde{H})}.$$
We will summarize the last paragraph in the following
proposition:
\begin{proposition}
If $\Pi_\G$ is a two step market evaluation then
$\tilde{\Pi}_\G$ is a two step market evaluation as well.
Furthermore, the penalty function of $\tilde{\Pi}_{\F^S}$ is
given by the penalty function of $\Pi_{\F^S}$ converted into
units of stock $i$.
\end{proposition}
Note that for evaluations which are not two step an easy
conversion of the penalty function as given in the last
proposition usually only works if the numeraire is
deterministic.

It is not hard to see that a two step market valuation is
\emph{always} market-consistent, see the theorems below. On the
other hand our Example \ref{example1} below shows that a
market-consistent valuation is not necessarily a two step
market valuation.

\subsection{An axiomatic characterization of two step market evaluations}
Notice that the Examples \ref{exdc2} all satisfy the
\emph{market local property}, i.e.:
\begin{itemize}
\item For every $A\in \F^S$ and $H\in L^\infty_+(\F)$ \be
    \label{localone} \Pi_{\G}(H)=\Pi_{\G}(I_A
    H)+\Pi_{\G}(I_{A^c}H). \ee
\end{itemize}
The binary options $I_A H$ and $I_{A^c}H$ can be seen as
insurances against the events $A$ or $A^c$ respectively. For
example if $A$ happens then the owner of the option $I_A H$
gets a non-negative amount $H$ which possibly depends on the
insurance process. If the evaluation $\Pi_\G$ is assumed to be
sublinear, the value of the insurance contract $I_A H$
\emph{plus} the value of the insurance contract $I_{A^c}H$
should be larger than the value of $H=I_A H+I_{A^c}H$. The
economic reason is that if the valuation of $H$ is decomposed
into the sum of the evaluations of the binary insurance
contracts then the additional uncertainty given by the event
$A$ for each binary contract, should lead to an increase of the
total insurance premium. The local property however postulates
that the uncertainty added to the payoff $H$ by an event from
the financial market should not carry any extra premium. We
will see in Section \ref{timecon} that, in a setting with a
correlated stock and insurance process, the local property is
satisfied for a time-consistent and market-consistent
evaluation if the stock process in continuous and the value of
the insurance process for is revealed at fixed time instances
(or more generally at predictable stopping times).

Now the question in which we are interested in the remainder of
this chapter, is the following: Given a conditional evaluation
$\Pi_{\G}$ does then its market-consistency and the market
local property imply that every position $H$ \emph{has to be
priced} with a two step market valuation? Our results below
will actually show that this statement holds.


The following theorem shows that market-consistency and the
market local property in the coherent case are equivalent to
two step market evaluations. Furthermore, it gives an explicit
formula for $\Pi_{\F^S}.$
\begin{theorem}
\label{main} The following statements are equivalent:
\begin{itemize}
\item[(i)]$\Pi_{\G}$ is a coherent market-consistent
    $\G$-conditional evaluation which satisfies the market
    local property.
\item[(ii)] There exists an $\F^S$-conditional coherent
    evaluation $\Pi_{\F^S}:L^\infty(\F) \to L^\infty(\F^S)$
    such that $\Pi_{\G}(H)=\EQG{\Pi_{\F^S}(H)}.$
    Furthermore, $\Pi_{\F^S}(H)=\esssup_{Z\in M'}\EFS{ZH},$
   with \be \label{M}
M':=\Big(\frac{d\mathbb{Q}_{\G}}{d\mathbb{P}_{\G}}\Big)^{-1}M\subset
\{Z\in L^1(\F)|\EGS{Z}=1\}=\mathcal{Q}_{\F^S}, \ee
\end{itemize}
and $M$ given by (\ref{coherent}).
\end{theorem}
The assumptions of the theorem are satisfied for all our
examples above except for the Mean-Variance and the Exponential
market-consistent principles. The reason is that these do not
satisfy positive homogeneity. However, in the case that the
filtration $\F^S$ is generated by countably many sets, and
$\Pi_\G$ is continuous, and monotone or $p$-norm bounded, we
can prove that an evaluation has to be a two step market
evaluation \emph{without the assumption of positive
homogeneity.}
\begin{theorem}
\label{main2}
 Suppose that $\F^S$ is generated by countably many sets.
Then the following statements are equivalent:
\begin{itemize}
\item[(i)] $\Pi_{\G}$ is a monotone, continuous
    market-consistent $\G$-conditional evaluation which
    satisfies the market local property.
\item[(ii)] There exists a monotone, continuous
    $\F^S$-conditional evaluation
    $\Pi_{\F^S}:L^\infty(\F)\to L^\infty(\F^S)$ such that
$$\Pi_{\G}(H)=\EQG{\Pi_{\F^S}(H)}.$$
Furthermore, if $\Pi_\G$ is additionally assumed to be
$p$-norm bounded then the statement also holds without the
monotonicity assumption in (i) and (ii).
\end{itemize}
\end{theorem}
The Mean-Variance principle is not monotone but it is $p$-norm
bounded. In particular, Theorems \ref{main}-\ref{main2} include
all our examples. Theorem \ref{main2} further applies to the
Exponential principle, and the Average Value at Risk and the
Semi-Deviation principle if $\l$ and $\delta$ only take values
between $0$ and $1$.

The proof of the last theorem heavily relies on the assumed
continuity of our evaluations which was not needed in Theorem
\ref{main}. The evaluation $\Pi_{\F^S}$ will be obtained only
as an a.s. limit using the martingale convergence theorem
without an explicit formula.

Note that the market-consistency of an evaluation $\Pi$ does
not necessarily imply that we get a representation as in
Theorem \ref{main} and Theorem \ref{main2}, since there are
market-consistent evaluations not satisfying the market local
property. This can be seen from the following example:
\begin{Example}
\label{example1} Assume that $\G$ is trivial and let $Z_1$ and
$Z_2$ be densities independent of $S$ with $Z_1\ne Z_2$. Then
$Z_i$ are also independent of
$\frac{d\mathbb{Q}}{d\mathbb{P}}$. Now suppose that the agent
is not sure if he should trust the density
$\frac{d\mathbb{Q}}{d\mathbb{P}}Z_1$ or
$\frac{d\mathbb{Q}}{d\mathbb{P}}Z_2.$ Therefore, he decides to
take a worst-case approach over all convex combinations of
$\frac{d\mathbb{Q}}{d\mathbb{P}}Z_1$ or
$\frac{d\mathbb{Q}}{d\mathbb{P}}Z_2.$ That is \be
\label{counter} \Pi(H):=\max_{\bar{\mathbb{P}}\in
M}\Ebar{H}=\max_{i=1,2}\E{\frac{d\mathbb{Q}}{d\mathbb{P}}Z_i
H}, \ee with $M=\{\l
\frac{d\mathbb{Q}}{d\mathbb{P}}Z_1+(1-\l)\frac{d\mathbb{Q}}{d\mathbb{P}}Z_2|\l\in[0,1]\}.$
It is straight-forward to check using the independence of $Z_i$
and $S$ that $\Pi$ is a market-consistent, coherent evaluation.
Let $H\geq 0$ be $\F$- but not $\F^S$-measurable. Assume
without loss of generality that the maximum in (\ref{counter})
is attained in $i=1.$ Now choose a set $A\in \F^S$ such that
\be \label{Z2} \E{ \frac{d\mathbb{Q}}{d\mathbb{P}}Z_2 I_A
H}>\E{ \frac{d\mathbb{Q}}{d\mathbb{P}}Z_1 I_A H}.\ee Since the
maximum in (\ref{counter}) is attained in $i=1$ we must have
that
 $$\E{ \frac{d\mathbb{Q}}{d\mathbb{P}}Z_2 I_{A^c} H}<\E{ \frac{d\mathbb{Q}}{d\mathbb{P}}Z_1 I_{A^c} H}.$$
 But then we get
\begin{align}
\label{Z1}
\Pi(HI_A)+\Pi(HI_{A^c})&=\E{ \frac{d\mathbb{Q}}{d\mathbb{P}}Z_2 I_A H}+\E{ \frac{d\mathbb{Q}}{d\mathbb{P}}Z_1 I_{A^c} H}\nonumber\\
&>
\E{ \frac{d\mathbb{Q}}{d\mathbb{P}}Z_1 I_A H}+\E{ \frac{d\mathbb{Q}}{d\mathbb{P}}Z_1 I_{A^c} H}=\Pi(H).
\end{align}
One possible choice for $A$ and $H$ would be to define $H=I_A
I_{\{Z_2>Z_1\}}+I_{A^c} I_{\{Z_1>Z_2\}}.$ Then (\ref{Z2})
always holds if $A$ is chosen to be a non-zero set with
$A\subset \{Z_2> Z_1\}$. Furthermore, one can ensure that
$\Pi(H)=\E{Z_1 H}$ by choosing $A$ such that $\mathbb{Q}(A^c)$
is sufficiently close to one. In particular, (\ref{Z1}) holds
and therefore $\Pi$ defined in (\ref{counter}) does not satisfy
the market local property. By the direction
(ii)$\Rightarrow$(i) of Theorem \ref{main} this implies that
$\Pi$ is not a two step market evaluation.
\end{Example}
\setcounter{equation}{0}
\section{Dynamic Evaluations in continuous time: Time-consistency and Market-consistency}
\label{timecon} In this section we will give sufficient
conditions in which the market local property, holds in a
\emph{dynamic} setting. Specifically we will obtain that in
certain dynamic settings time-consistency and
market-consistency imply that all evaluations have to be two
step market evaluations. Time-consistency in a dynamic setting
often has strong implications. For instance, for general
preferences the indifference price of an agent with
time-consistent dynamic preferences are recursive if and only
if the preferences are cash-invariant, see Theorem 3.4,
\citet{cheridito2006time}.

Subsequently, we fix a finite time horizon $T>0.$ Throughout
the rest of the paper we assume that additional to the stocks
$(S_t)_{0\leq t\leq T}=((S^1_t,\ldots,S^n_t))_{0\leq t\leq T}$,
we have an untraded insurance process $(Y_t)_{0\leq t\leq T}$.
For the sake of simplicity let us assume that the insurance
process is one-dimensional. (The generalization is
straight-forward.) Let $\bar{\F}^S$ be filtration generated by
$S$, and let $\bar{\F}^Y$ be the filtration generated by $Y$.
We again assume that the financial market is complete while the
entire market is incomplete. Denote by $\mathbb{Q}$ the unique
martingale measure on $\bar{\F}^S_T$ with density
$\frac{d\mathbb{Q}}{d\mathbb{P}}.$

Define the total information which is available as
$\F:=\bar{\F}^S\vee \bar{\F}^Y:=\s(\bar{\F}^S\cup \bar{\F}^Y).$
Setting $\mathbb{Q}(A):=\E{\frac{d\mathbb{Q}}{d\mathbb{P}}I_A}$
for $A\in\F_T$, we can extend $\mathbb{Q}$ canonically to the
whole filtration. We call a collection of mappings
$(\Pi_{\s})_{0\leq \s\leq T},$ a continuous-time \emph{dynamic
evaluation} if it has the following properties:
\begin{itemize}
\item For all stopping times $\s$, $\Pi_\s$ is an
    $\F_\s$-conditional evaluation.
\item \emph{Time-Consistency:} For every $H\in
    L^\infty(\F_T)$: $$ \Pi_{\s}(H)=
    \Pi_{\s}(\Pi_\t(H))\mbox{ for all stopping times
    }\s\leq \t.$$
\end{itemize}
In a dynamical context time-consistency is a natural assumption
to glue together static risk measures. It means that the same
risk is assigned to a financial position regardless of whether
it is calculated over two time periods at once or in two steps
backwards in time. For a general analysis of weaker notions of
time-consistency see e.g. \citet{roorda2005coherent}.
\begin{Remark}
Alternative names for our definition of time-consistency
would have been `recursiveness' or `tower property'. Note that
in the literature often the following notion of
time-consistency is used: if an asset $H_1$ is preferred to an
asset $H_2$ under all possible scenarios at some time $\t$ then
$H_1$ should also have been preferred at every time $\s$ before
$ \t.$ Let us call the latter definition of time-consistency
property (TC). Now if $\Pi$ is monotone it is well known that our notion of time-consistency
is in fact equivalent to property (TC). Furthermore, since
(TC) implies the monotonicity of $\Pi$, our notion of
time-consistency includes (TC) as a special case, but it can
also be applied to non-monotone evaluations.
\end{Remark}

 Time-consistent evaluations have been discussed in
continuous time by \citet{peng2004filtration},
\citet{frittelli2004dynamic}, \citet{gianin2006risk},
\citet{Delb2006mStableSets},
\citet{KlSchwe:2007:Dynamicindifferencevaluation},
\citet{bion2008dynamic}, \citet{bion2009time}, and
\citet{barrieu:elkaroui:2009:carmona}.
\citet{DuffiEpste:1992:Stochasticdifferentialutility},
\citet{ChenEpste:2002:Ambiguityriskand}, and
\citet{maccheroni2006dynamic} deal with dynamic preferences
using similar notions of time-consistency.

 Given a dynamic evaluation $(\Pi_{\s})$
we define $\Pi_{\s,\t}$ to be equal to $\Pi_\s$ restricted to
$L^\infty(\F_\t),$ i.e., $\Pi_{\s,\t}=\Pi_\s |L^\infty(\F_\t).$
Next we will assume that the insurance process $Y$ a.s.
\emph{jumps only at finitely many predictable time instances},
say $0\leq \tau_1\leq \tau_2\leq \ldots$ and there is no
additional randomness added between the jumps, i.e.,
$\bar{\F}^Y_{\t_{i+1}-}=\bar{\F}^Y_{\t_{i}}$ for all $i$.
Recall that a stopping time $\t$ is predictable if there exists
a sequence of stopping times $\tau^n<\t$ such that
$\t^n\uparrow \t.$ Hitting times of continuous processes are
predictable while jump times of L\'evy processes (or more
general, strong Markov Feller processes) are not. On the set
where $Y$ does not jump at all we set
$\t_1=T.$\footnote{Similarly we can set some $\t_i(\omega)$
equal to $T$ if not all paths, $(Y_t(\omega))_t$, have the same
number of jumps.} One example could be given by a setting in
which the insurance process $Y$ is only updated at finitely
many \emph{fixed} time instances, $t_1<\ldots<t_k.$ Another
possibility could be that damages occur at \emph{unpredictable}
stopping times $\bar{\t}_i,$ but the amount of money the
insurance will have to pay is not clear right away. Instead the
insurance needs some additional time, say $\varepsilon>0$, to
agree to a certain amount and to pay it out at
$\bar{\t}_i+\epsilon$, respectively.

We will use the following definition in a dynamic setting.
\begin{definition}
We will say that a conditional continuous-time evaluation
$(\Pi_\s)_{\s\in [0,T]}$ is market-consistent if for every
stopping time $\s$ and every financial payoffs $H^S\in
L^\infty(\bar{\F}^S_T\vee \F_\s)$ and $H\in L^\infty(\F_T)$
$$\Pi_{\s}(H^S+H)=\EQsigma{H^S}+\Pi_{\s}(H).$$
\end{definition}
Note that this definition of market-consistency coincides with
the definition in the static case with $\G=\F_\s$ and
$\bar{\F}^S=\bar{\F}^S_T.$

Now for every stopping time $\s$ we define $\t_\s$  to be the
time of the next jump after $\s$, i.e., $\t_\s:=\inf\{t>
\s|\Delta Y_t>0\}\wedge T.$ Denote
$\F^S_{\t_\s}:=\bar{\F}^S_{\t_\s}\vee \F_{\s}.$ That is
$\F^S_{\t_\s}$ is the $\s$-algebra which includes all the
information (of both stock and insurance process) up to time
$\s$ and additionally the information of the stock process up
to the time of the next jump of $Y$.
\begin{theorem}
\label{predictable} Suppose that $S$ is continuous and the
insurance process is as described above. Let $(\Pi_\s)_{0\leq
\s\leq T}$ be a time-consistent and market-consistent
evaluation such that for every $\s$, $\Pi_\s(\cdot)$ is
continuous. Then for every stopping time $\s$, we have that
$\Pi_{\s,\t_\s}$ satisfies the market-local property. In
particular, if $(\Pi_\s)$ is additionally assumed to be either
monotone, $p$-norm bounded, or positively homogeneous then for
every stopping time $\s$ there exists an
$\F^S_{\t_\s}$-conditional evaluation
$\Pi_{\F^S_{\t_\s}}:L^\infty(\F_{\t_\s})\to
L^\infty(\F^S_{\t_\s})$ such that
$$\Pi_{\s,\t_\s}(H)=\EQsigma{\Pi_{\F^S_{\t_\s}}(H)}.$$
\end{theorem}
This theorem shows in particular that in a setting where the
agent just observes the insurance process at finitely many time
instances, \emph{every} market-consistent and time-consistent
evaluation has to admit a representation of the form
 $\EQsigma{\Pi_{\F^S_{\t_\s}}(H)}$ at every stopping time $\s.$
 In other words, an agent who wants to use a time-consistent and market-consistent evaluation
 \emph{has to} apply a two step market evaluation.

Theorem \ref{predictable} also yields the following corollary:
\begin{corollary}
\label{predictable2} Suppose that $S$ is continuous, and that
the insurance process $Y$ a.s. jumps only at finitely many
fixed time instances, say $0\leq t_1\leq t_2\leq \ldots \leq
t_k,$ and $\bar{\F}^Y_{t_{i+1}-}=\bar{\F}^Y_{t_i}.$ If
$(\Pi_\s)_{0\leq \s\leq T}$ is a time-consistent and
market-consistent evaluation which is either monotone, $p$-norm
bounded or positively homogeneous then for every $s\in[0,T]$
there exists an $\F^S_{t_i}$-conditional evaluation
$\Pi_{\F^S_{t_i}}:L^\infty(\F_{t_i})\to L^\infty(\F^S_{t_i})$
such that
$$\Pi_{s,t_{i}}(H)=\EQs{\Pi_{\F^S_{t_i}}(H)},$$
where $t_i$ is the next time instance after $s$ at which the
insurance process jumps.
\end{corollary}
We have restricted our analysis to payoffs rather than payment
streams. However, if $H$ is a discrete payment streams which
pays cash amounts $(H_{\eta_1},\ldots,H_{\eta_k})$ at stopping
times $\eta_1<\eta_2<\ldots<\eta_k$, we could consider
evaluations $\Pi_\s$ mapping payment streams, starting paying
amounts from time $\s$ on, to $L^\infty(\F_\s).$ In this case
time-consistency could be defined as \beas \Pi_{\s}(H)=
\Pi_{\s}(H I_{[0,\t)}+\Pi_\t(H I_{[\t,T]})) \mbox{ for all }
\s\leq \t, \eeas see also \citet{cher:delb:kupp:2006:dynamic},
or \citet{jobert:rogers:2008:valuations}.
By a proof analogue
to the one for Theorem \ref{predictable} one can  then show
that for an evaluation satisfying  similar properties as above
time-consistency and market-consistency entail
$$\Pi_{\s,\t_\s}(H I_{[\s,\t_\s]})=\EQsigma{\Pi_{\F^S_{\t_\s}}(HI_{[\s,\t_\s]})}=\EQsigma{\Pi_{\F^S_{\t_\s}}(H_{\t_\s})}+  \EQsigma{ \sum_{\s \le \eta < \t_\s} H_\eta}.$$

\setcounter{equation}{0}
\section{Market-consistent evaluations in continuous time}
\subsection{Results on market-consistent BSDEs}
\label{g} In a continuous time Brownian-Poisson setting we can
provide examples of time-consistent and market-consistent
evaluations by $g$-expectations. It is well known that
$g$-expectations induce time-consistent evaluations, see for
instance \citet{peng2004filtration},
\citet{frittelli2004dynamic}, \citet{gianin2006risk},
\citet{bion2008dynamic}, \citet{barrieu:elkaroui:2009:carmona},
or \citet{el2009cash}. In this section we will give a complete
characterization of $g$-expectations which are
market-consistent.

Suppose the filtration $\F$ is generated by the following
independent processes: an $n$-dimensional standard Brownian
motion $W^f$, a $d$-dimensional standard Brownian motion $W$,
and a Poisson random counting measure $N(ds,dx)$ defined on
$[0,T]\times\mathbb{R}\setminus\{0\}.$ We denote the
corresponding compensator by
$$ \hat{N}(ds,dx)=\nu(dx)ds. $$
We assume that the measure $\nu(dx)$ is non-negative and
satisfies
$$\int_{\mathbb{R}\setminus\{0\}} (|x|^2\wedge 1) \nu(dx)<\infty.$$
Denote $\tilde{N}(ds,dx):=N(ds,dx)-\hat{N}(ds,dx).$

Suppose that we have a bond $B$ with $B_0=1$ and $dB_t=rB_tdt.$
We assume that the stock process $S=(S^1,\ldots,S^n)$ is given
by
$$dS^i_t=S^i_t(\mu^i(t,S_t)dt+\tilde{\s}^i(t,S_t)dW^f_t), \,\,\,\,\,S^i_0=s^i_0,$$
with $s^i_0>0$ and $\mu^i:[0,T]\times \mathbb{R}^n\to
\mathbb{R}$ and $\tilde{\s}^i:[0,T]\times \mathbb{R}^n\to
\mathbb{R}^n$ for $i=1,\ldots,n.$ Note that we used vector
notation for the stochastic integral. Define the vector
$\mu=(\mu^i)_{i=1,\ldots,n}$ taking values in $\mathbb{R}^n.$
The rows of the matrix $\tilde{\sigma}:[0,T]\times
\mathbb{R}^n\to \mathbb{R}^{n\times n}$ should be given by
$(\tilde{\sigma}^i)_{i=1,\ldots,n}.$ We assume that $\mu$ and
$\tilde{\s}$ are uniformly bounded. Furthermore, $\tilde{\s}$
should be invertible, and uniformly elliptic, i.e., there
exists $K_1,K_2>0$ such that $K_1 I_n\preceq \tilde{\s}
\tilde{\s}^\intercal \preceq K_2 I_n.$ Furthermore, we need
standard measurability and Lipschitz continuity assumptions on
$\mu$ and $\tilde{\s}$. Then it is well known that a solution
$S$ for the SDE above exists and that the corresponding
financial market, consisting of $(S,B)$, is complete. Generally
payoffs can depend on $(W^f,W,\tilde{N})$ and may not be
replicable.

Let $\bar{\F}^S$ be the filtration generated by $S.$ Denote the
negative market price of uncertainty by
$\theta_t=-\tilde{\s}^{-1}(t,S_t)(\mu(t,S_t)-r e),$ where $e$ is
an $n$-dimensional vector consisting of ones. As in the
sections before we will denote by $H$ hedgeable and unhedgeable
discounted payoffs.

We will consider evaluations of $H$ given by the solutions of
backward stochastic differential equations (BSDEs). Denote by
$\mathcal{P}$ the predictable $\s$-algebra on the entire
filtration. Let
$$\mathcal{H}^2_m:=\Big\{Z=(Z^1,\ldots,Z^m)\in\mathcal{P}\Big|\E{\int_0^T |Z_s|^2ds}<\infty\Big\},$$
where we denote by $|\cdot|$ the Euclidean norm.

Let $\mathcal{S}^2$ be the space of all one-dimensional
optional processes whose path-maximum is square integrable with
respect to $\mathbb{P}.$ Let $L^2(\nu(dx))$ be the space all
$\mathcal{B}(\mathbb{R}\setminus\{0\})$-measurable functions
mapping from $\mathbb{R}\setminus\{0\}$ to $\mathbb{R}$, which
are square integrable with respect to $\nu$, where, as usual,
two functions are identified if they are equal $\nu$ a.s.
Define $L^2(\nu(dx)\times d\mathbb{P}\times ds )$ as all
$\mathcal{P}\otimes
\mathbb{B}(\mathbb{R}\setminus\{0\})$-measurable functions
which are square-integrable with respect to $\nu(dx)\times
d\mathbb{P}\times ds.$ Now suppose that we have a suitably
measurable function $g:[0,T]\times
\Omega\times\mathbb{R}^2\times L^2(\nu(dx))\to \mathbb{R}.$

A solution of the BSDE with driver $g(t,\omega
,z^f,z,\tilde{z})$ and terminal condition $H\in L^\infty(\F_T)$
is a quadruple  of processes $(Y(H),Z^f,Z,\tilde{Z})\in
\mathcal{S}^2\times \mathcal{H}^2_n\times\mathcal{H}^2_d\times
L^2(\nu(dx)\times  d\mathbb{P}\times ds)$ such that \beas
dY_t(H)=-g(t,Z^f_t,Z_t,\tilde{Z}_t)dt+Z^f_t dW^f_t+Z_t
dW_t+\int_{\mathbb{R}\setminus\{0\}}\tilde{Z}_t(x)
\tilde{N}(dt,dx)\mbox{ and }Y_T(H)=H.\eeas

Often BSDEs are also written in the following equivalent form:
\beas Y_t(H)= H+\int_t^T g(s,Z^f_s,Z_s,\tilde{Z}_s)ds-\int_t^T
Z^f_s dW^f_s-\int_t^T Z_s
dW_s-\int_t^T\int_{\mathbb{R}\setminus\{0\}} \tilde{Z}_s(x)
\tilde{N}(ds,dx).\eeas Since the terminal condition is given at
maturity time $T$, BSDEs have to be computed backwards in time.
As in many applications a terminal reward is specified and
solutions of BSDEs satisfy a dynamic programming principle,
BSDEs are often applied to solve problems in stochastic optimal
control and mathematical finance, see apart from the papers
mentioned above for instance
\citet{KarouiPengQuenez1997BSDE-Finance},
\citet{lazrak2003generalized}, \citet{hamadene2007starting},
and \citet{porchet2008valuation}; or for the discrete time case
\citet{madan2010valuation}.

Subsequently, we will always assume that the BSDE which we
consider has a unique solution. This is for instance the case
if $g(t,0,0,0)$ is in $L^2(d\mathbb{P}\times dt)$, and $g$ is
uniformly Lipschitz continuous; that is there exists $K>0$ such
that
$$|g(t,z^f_1,z_1,\tilde{z}_1)-g(t,z^f_0,z_0,\tilde{z}_0)|\leq
K\bigg(z^f_1-z^f_0|+|z_1-z_0|+\sqrt{\int_{\mathbb{R}\setminus\{0\}}
|\tilde{z}_1(x)-\tilde{z}_0(x)|^2\nu(dx)}\bigg),$$ see for instance \citet{royer2006backward} and the reference therein.

\begin{Example} Let $H$ be a bounded payoff and define
$Y_t(H)=\Et{H}.$ Then by the martingale representation theorem
(see e.g. \citet{jacod1987limit}, Sec. 3, Theorem 4.34) there
exist predictable square-integrable processes $Z^f,Z$ and
$\tilde{Z}$ such that $Y$ satisfies
$$dY_t(H)=Z^f_t dW^f_t+Z_t dW_t+\int_{\mathbb{R}\setminus\{0\}}\tilde{Z}_t(x) \tilde{N}(dt,dx)\mbox{ and
}Y_T(H)=H.$$
This is the simplest BSDE with $g=0$.
\end{Example}

Hence, a conditional expectation may be seen as a BSDE with
$g=0$. This is why BSDEs are being referred to as
$g$-expectations. The name should express that a BSDE may be
viewed as generalized (usually non-linear) conditional
expectation with an additional drift.

\begin{Example} Let $H$ be a bounded payoff and define
$Y_t(H)=\EQt{H}.$ Then by the martingale representation theorem
and by the Girsanov theorem
 $Y(H)$ satisfies
$$dY_t(H)= -\theta_t Z^f_t dt + Z^f_t dW^f_t+Z_t dW_t+\int_{\mathbb{R}\setminus\{0\}}\tilde{Z}_t(x) \tilde{N}(dt,dx)\mbox{ and
}Y_T(H)=H.$$
This is a linear BSDE with $g(t,z^f,z,\tilde{z})=\theta_t z^f$.
\end{Example}

Subsequently, we will write
$$\mathcal{E}^g_t(H)=Y_t(H).$$
In a Markovian setting $g$-expectations correspond to
semi-linear parabolic PDEs (or PIDEs in the case of jumps), see
for instance \citet{KarouiPengQuenez1997BSDE-Finance} in a
Brownian setting (see \citet{barles1997backward} in the case of
jumps).

It may be seen that if $g$ is convex and $g(t,0,0,0)=0$ then
the evaluation defined by
$$\Pi_\s(H):=\mathcal{E}^g_\s(H),$$
is normal, monotone, cash invariant, convex, time-consistent
and satisfies the local property. Hence, $g$-expectations give
us an abundance of time-consistent, continuous-time
evaluations. There are also certain sufficient conditions under
which, in a Brownian filtration, a time-consistent evaluation
is induced as the solution of a $g$-expectation, see
\citet{coquet2002filtration}.

The following theorem gives a complete characterization of
market-consistent evaluations given by $g$-expectations:
\begin{theorem}
\label{g-expectation} A $g$-expectation is market-consistent if
and only if $g(t,z^f,z,\tilde{z})-\theta_t z^f$ does not depend
on $z^f$ $d\mathbb{P}\times dt$ a.s.
\end{theorem}
\subsection{Examples of market-consistent BSDEs}
\label{g2}
 To get an interpretation of $g$ we will consider
some examples. The next proposition gives a \emph{dynamic}
market-consistent extension of the exponential principle. It
follows directly from Theorem \ref{g-expectation} above and
Theorem 2 in \cite{morlais2010new}. (\cite{morlais2010new} also
allows for a hedging set $\mathcal{C}$ and an initial capital
amount $x$. In our case this becomes $\mathcal{C}=\{0\}$ and
$x=0$.)
\begin{proposition}
\label{exp} Define the evaluation $\Pi$ as the solution of the BSDE
with driver function
$$g(t,z^f,z,\tilde{z})=\theta_t z^f+\frac{1}{2\g}|z|^2 +
 \g\int_{\mathbb{R}\setminus\{0\}}[\exp\{\frac{\tilde{z}(x)}{\g}\}-\frac{\tilde{z}(x)}{\g}-1]\nu(dx).$$
Then (i) $\Pi$ is market-consistent, and (ii) for pure insurance risk (i.e., terminal conditions independent of $S$) $\Pi$ corresponds to the exponential principle from Example \ref{exdc}.
\end{proposition}
Other examples of the driver function $g$ can be obtained by
looking at one-period evaluations in discrete time defined
recursively. Namely, suppose that we have an equi-spaced time
grid $I=\{0,h,2h,\ldots,T\}$ where we assumed without loss of
generality that $T$ is a multiple of $h.$ The filtration
$(\F_{ih})_{i=0,1,\ldots,T/h}$ is generated by
 $(W^f_{ih},W_{ih},\tilde{N}((0,ih],dx))_{i=1,\ldots,T/h}.$
Define $S^{h,j}_0=s^j_0,$
$S^{h,j}_{(i+1)h}=S^{h,j}_{ih}(1+\mu^j(ih,S^h_{ih})h+\tilde{\s}^j(ih,S^h_{ih})\Delta
W^{f}_{(i+1)h})$ for $i=1,\ldots,T/h$ and $j=1,\ldots,n$, and
$S^h=(S^{h,1},\ldots,S^{h,n}).$ Denote further
$\F^{S^h}_{(i+1)h}=\bar{\F}^{S^h}_{(i+1)h}\vee \F_{ih}$, where
$\bar{\F}^{S^h}$ is the filtration generated by $S^h.$ In other
words $\F^{S^h}$ is the information of the (discrete-time)
stock process together with the previous values of the
insurance process. Let $\mathbb{Q}^h$ be the measure (with
$\bar{\F}^{S^h}$-measurable density) such that $ \Delta
W^{f,*}_{(i+1)h}:= \Delta W^{f}_{(i+1)h}-\theta_{ih}h$ is a
martingale.

Now we can use evaluations from our Examples \ref{exdc2} over
one period and glue them together recursively. Using our two
step procedures for the one-periodic evaluations could be
natural, in particular, if the stock process can be observed
before the insurance process. Suppose for instance for a moment
that we are at time $ih$ and the stock process, $S^h,$ can be observed
at time $(i+\frac{1}{2})h,$ whereas the insurance process is
revealed after the stock process at time $(i+1)h.$ Since data
from the financial market can be observed almost continuously,
while data from insurance companies are typically observed less
frequently, this may not be an unreasonable assumption, see
also Section \ref{timecon}. Of course our insurance process
will possibly be effected by the financial market through its
correlation to $W^f$. However, it will not be completely
predictable due to its dependance on the jumps and $W$. In this
situation if an evaluation $(\Pi_{\s})_{\s\in[0,T]}$ is
time-consistent, then market-consistency would imply that for
the aggregated evaluation $(\Pi_{ih})_{i=0,1,\ldots,T/h}$ we
have
$$\Pi_{ih}(H_{(i+1)h})=\Pi_{ih}(\Pi_{(i+\frac{1}{2})h}(H_{(i+1)h}))=\EQti{\Pi_{(i+\frac{1}{2})h}(H_{(i+1)h})}.$$
The last equation holds by market-consistency since
$\Pi_{(i+\frac{1}{2})h}(H_{(i+1)h}),$ is
$\F^{S^h}_{(i+\frac{1}{2})h}=\bar{\F}^{S^h}_{(i+\frac{1}{2})h}\vee
\F_{ih}$-measurable. Therefore, the setting outlined above
would indeed give rise to applying a two step market evaluation
at time-instances $ih$, to evaluate payoffs up to time
$(i+1)h$. We will calculate one example explicitly by
considering at time instances $ih$ the mean-variance two step
market evaluation
$$\Pi^v_{ih,(i+1)h}(H_{(i+1)h})=\EQh{\Enip{H_{(i+1)h}}+\frac{1}{2}\alpha
Var_{\F^{S^h}_{(i+1)h}}[H_{(i+1)h}]}, \ \ \a\geq 0
.$$

To obtain a multi-period evaluation
$(\Pi^h_{ih})_{i=0,1,\ldots,T/h}$ on the whole filtration with
$\Pi^h_{ih}:L^\infty(\F_{T})\to L^\infty(\F_{ih})$ we will
define $\Pi^h$ recursively by setting
\begin{align}
\label{rec}
\Pi^h_T(H)=H \mbox{ and }\Pi^h_{ih}(C)=\Pi^v_{\F_{ih}}(\Pi^h_{(i+1)h}(H))\mbox{ for }i=0,1,\ldots,T/h-1.
\end{align}
The following proposition is proved in the appendix:
\begin{proposition}
\label{mean-var} Suppose that the evaluation $\Pi^h$ is
constructed by (\ref{rec}), i.e., using locally the
Mean-Variance two step market evaluation. Then for every
terminal payoff $H\in L^\infty(\F_T)$ there exists predictable
$(Z^{h,f},Z^h,\tilde{Z}^h)$ and a martingale $(L^h_{ih})_i$
orthogonal (under $\mathbb{Q}^h$) to\newline
$(W^f_{ih},W_{ih},\tilde{N}((0,ih],dx))_i$ such that for all
$i$ we have
\begin{align}
\label{discretedc} \Pi^h_{ih}(H)&=H+\sum_{j=i}^{T/h-1}\bigg( \Big[\theta_{jh}Z^{h,f}_{jh}+\frac{\a}{2}\Big(|Z^{h}_{jh}|^2+
\int_{\mathbb{R}\setminus\{0\}}|\tilde{Z}^h_{jh}(x)|^2\nu(dx)\Big)\Big]h
\nonumber\\
&\hspace{0.1cm}+\frac{\a}{2}\EQj{\Big(\Delta L^h_{(j+1)h}-{\rm E}_{\F^{S^h}_{(j+1)h}}\Big[\Delta
L^h_{(j+1)h}\Big]\Big)^2}\bigg)
-\sum_{j=i}^{T/h-1} Z^{h,f}_{jh} \Delta W^f_{(j+1)h}
\nonumber\\
&\hspace{0.1cm}
-\sum_{j=i}^{T/h-1} Z^h_{jh} \Delta W_{(j+1)h}
-\sum_{j=i}^{T/h-1}\int_{\mathbb{R}\setminus\{0\}}
\tilde{Z}^h_{jh}(x) \tilde{N}((jh,(j+1)h],dx)-(L^h_T-L^h_{ih}).
\end{align}
In particular, $\Pi^h$ satisfies a discrete-time BSDE.
\end{proposition}
From Proposition \ref{mean-var} we may infer that
\begin{align*}
\Ejh{\Delta
\Pi^h_{(j+1)h}(H)}&=-\Big[\theta_{jh}Z^{h,f}_{jh}+\frac{\a}{2}\Big(|Z^{h}_{jh}|^2+
\int_{\mathbb{R}\setminus\{0\}}|\tilde{Z}^h_{jh}(x)|^2\nu(dx)\Big)\Big]h\\
&\hspace{1.5cm}
-\frac{\a}{2}\EQj{\Big(\Delta L^h_{(j+1)h}-{\rm E}_{\F^{S^h}_{(j+1)h}}\Big[\Delta
L^h_{(j+1)h}\Big]\Big)^2}.
\end{align*}
 Note that
the orthogonal martingale terms $\Delta L^h_{(j+1)h}$ arise
because the discretized Brownian Motions do not have the
representation property. However, the continuous time Brownian
motions and the Poisson random measure do have the
representation property. Therefore, if we ignore the $L^h$ then
an analogous infinitesimal way of charging the risk in
continuous time would be an evaluation which satisfies $$\Et{d
\Pi_{t}(H)}=-\Big[\theta_tZ^f_t+\frac{\a}{2}\Big(|Z_t|^2+\int_{\mathbb{R}\setminus\{0\}}|\tilde{Z}_t(x)|^2\nu(dx)\Big)\Big]dt.$$
This corresponds to an evaluation given by the solution of a
BSDE with driver function
 \beas
g(t,z,\tilde{z})=\theta_tz^f+\frac{\a}{2}\Big(|z|^2+\int_{\mathbb{R}\setminus\{0\}}|\tilde{z}(x)|^2\nu(dx)\Big).
\eeas
\begin{Remark}
The analogy stated above only corresponds to the local way of
charging risk. A global correspondence of charging risk would
involve proving a convergence result for the whole path. In the
                                           case that the
                                           driver function
                                           is Lipschitz
                                           continuous it
                                           can actually be shown in
                                           a purely
                                           Brownian setting
                                           that after an
                                           appropriate
                                           scaling the whole path of the discrete
                                           time evaluations
                                           converges to
                                           the corresponding
                                           solution of the
                                           BSDE;
                                           see
                                           \citet{stadje2010extending}.
                                           However, in
                                           Proposition 5.5.
                                           the driver
                                           function is
                                           quadratic in
                                           $z$. In this
                                           case already in
                                           the purely Brownian
                                           setting an
                                           extension by
                                           convergence may
                                           not always be possible. \citet{cheridito2010bs}
                                           give an example
                                           which
                                           shows that in a
                                           setting where
                                           the discrete
                                           time filtration
                                           is generated by
                                           a Bernoulli
                                           random walk it may
                                           happen
                                           that the
                                           discrete time
                                           $H^n$ are
                                           uniformly
                                           bounded and
                                           converge to
                                           $H\in
                                           L^\infty(\F)$
                                           but the discrete
                                           time evaluations
                                           explode.
\end{Remark}
\section{Summary \& Conclusions}
In this paper we have studied the extension of standard
actuarial principles in time-consistent and market-consistent
directions by introducing a new market-consistent evaluation
procedure which we call `two step market evaluation.' On the
one hand, two step market evaluations sometimes arise when an
agent starts with an evaluation that is not market-consistent,
such as the Average Value at Risk or the Exponential Premium
principle, and then engages in hedging. On the other hand,
market-consistent evaluations can also be defined directly by
applying a standard evaluation technique, conditional on the
stock process. In this case the structure of many standard
evaluation techniques can be preserved.

We have shown that two step market evaluations are invariant if
the stock is taken as a numeraire. In Theorem 3.8 and Theorem
3.9
 a complete axiomatic characterization for two step
market evaluations  is provided. Moreover, we have proved that
in a dynamic setting with a continuous stock prices process and
an insurance process being revealed at predictable times every
evaluation which is time-consistent and market-consistent is a
two step market evaluation up to the next predictable time,
which gives a strong argument for their use. We have also
characterized market-consistency in terms of $g$-expectations
and studied the extension of the mean-variance and the
exponential principle to continuous time, in a setting with
jumps. Our analysis shows that two step evaluations can provide
a useful, computationally tractable tool for market-consistent
valuations.

\appendix

\section{Appendix: Technical Material \& Proofs}
\subsection{Background to the Essential Supremum}
The first part of the appendix is basically a summary of the
definitions and results given in A.5
\citet{follmer2004stochastic}. Consider a family of random
variables $M$ on a given probability space
$(\Omega,\bar{\F},\mathbb{P}).$ Now if $M$ is countable then
$Z^*(\omega)=\sup_{Z\in M} Z(\omega)$ is also measurable.
However, measurability is not guaranteed if $M$ is uncountable.
Even if the pointwise supremum is measurable it might not be
the right concept when we focus on a.s. properties. For
instance if $\mathbb{P}$ is the Lesbegue measure on
$\Omega=[0,1]$ and $M=\{I_{\{x\}}|0\leq x\leq 1\}$ then
$\sup_{Z\in M}Z=1$ while $Z=0$ a.s. for all $Z\in M.$ This
suggest the following notion of an \emph{essential supremum}
defined in terms of almost sure inequalities. This result can
be found as Theorem A.32 in \citet{follmer2004stochastic}.
\begin{theorem}
\label{Schied} Let $M$ be any set of random variables on
$(\Omega,\bar{\F},\mathbb{P}).$
\begin{itemize}
\item[(a)] There exists a random variable $Z^*$ such that
$$Z^*\geq Z\,\,\, \mathbb{P}\mbox{-a.s. for all }Z\in M.\hspace{2cm}(A)$$
Moreover, $Z^*$ is a.s. unique in the following sense: Any
other random variable $\hat{Z}$ with property (A) satisfies
$\hat{Z}\geq Z^*$ $\mathbb{P}$-a.s.
\item[(b)] Suppose that $M$ is directed upwards, i.e., for
    $Z_1,Z_2\in M$ there exists $Z_3\in M$ with $Z_3\geq
    \max(Z_1,Z_2).$ Then there exists an increasing
    sequence $Z_1\leq Z_2\leq\ldots$ in $M$ such that
    $Z^*=\lim_n Z_n$ $\mathbb{P}$-a.s.
\end{itemize}
\end{theorem}
\begin{definition}
The random variable $Z^*$ in the the theorem above is called
the \emph{essential supremum} of $M$ and we write:
$$Z^*=\esssup_{Z\in M} Z.$$
We define the essential infimum similarly.
\end{definition}
If the probability space is finite the essential supremum
corresponds to the pointwise supremum taken over all atoms.
\begin{lemma}
\label{appendix} If $M$ satisfies the concatenation property,
i.e., for every $A\in \bar{\F}$ and $Z_1,Z_2\in  M$ we have
that $Z_1I_A+Z_2I_{A^c}\in M,$ then $M$ is directed upwards.
\end{lemma}
\begin{proof}
Define
$$Z^*=Z_1I_{\{Z_1\geq Z_2\}}+Z_2I_{\{Z_1< Z_2\}}.$$
By the concatenation property $Z^*\in M,$ and by definition
$Z^*\geq \max(Z_1,Z_2).$
\end{proof}
\subsection{Proofs of the results in Section \ref{markcon}}
\textit{Proof of Proposition \ref{prop}.} (i)$\Rightarrow$(ii)
By (\ref{duality}) we have that
$$
\Pi_{\G}(H)=\esssup_{\xi\in L^1(\F)} \{\EG{\xi H}-\Pi^*_{\G}(\xi)\}.$$
Furthermore,
\begin{align}
\label{restrdual}
\Pi^*_{\G}(\xi)&=\esssup_{H\in L^\infty(\F)} \{\EG{\xi H}-\Pi_{\G}(H)\}\nonumber\\
&\geq \esssup_{H\in L^\infty(\F^S)} \{\EG{\xi H}-\Pi_{\G}(H)\}
\nonumber\\
&= \esssup_{H\in L^\infty(\F^S)} \{\EG{\EGS{\xi} H}-\Pi_{\G}(H)\}.
\end{align}
The last term in (\ref{restrdual}) is the dual of $\Pi_{\G}$
restricted to $\F^S$ evaluated at $\EGS{\xi}$. Now by
assumption $\Pi_{\G}(H^S)=\EQG{H^S}$ is linear. Thus, its dual
penalty function,
$(\Pi_\G|\F^S)^*:L^1(\F^S)\to\mathbb{R}\cup\{\infty\},$ must be
equal to the indicator function which is zero if the input
argument is $\frac{d\mathbb{Q}_{\G}}{d\mathbb{P}_{\G}},$ and
infinity else. But then, by the inequality above,
$\Pi^*_{\G}(\xi)$ must be equal to infinity as well if
$\EGS{\xi}\ne \frac{d\mathbb{Q}_{\G}}{d\mathbb{P}_{\G}}.$ Thus,
it is sufficient to consider $\xi$ of the form
$\xi=\frac{d\mathbb{Q}_{\G}}{d\mathbb{P}_{\G}}Z$ for a $Z:=
\frac{d\mathbb{Q}_{\G}}{d\mathbb{P}_{\G}}\in
\mathcal{Q}_{\F^S}.$ Defining for $Z\in \mathcal{Q}_{\F^S},$
$c_\G(Z)=\Pi^*_\G(\frac{d\mathbb{Q}_{\G}}{d\mathbb{P}_{\G}}Z),$
we have indeed that
\begin{align*}
\Pi_{\G}(H)&=\esssup_{\xi\in L^1(\F)} \{\EG{\xi H}-\Pi^*_{\G}(\xi)\}\\
&
=\esssup_{Z\in \mathcal{Q}_{\F^S}} \{\EG{\frac{d\mathbb{Q}_{\G}}{d\mathbb{P}_{\G}}ZH}-\Pi^*_{\G}(\frac{d\mathbb{Q}_{\G}}{d\mathbb{P}_{\G}}Z)\}=
\esssup_{Z\in \mathcal{Q}_{\F^S}} \{\EQG{ZH}-c_\G(Z)\}.
\end{align*}

(ii)$\Rightarrow$(iii): It is by (ii)
\begin{align*}
\Pi_{\G}(H^S+H)&=\esssup_{Z\in \mathcal{Q}_{\F^S}} \Big\{\EG{\frac{d\mathbb{Q}_{\G}}{d\mathbb{P}_{\G}}Z (H^S+H)}-c_\G(Z)\Big\}\\
&=\esssup_{Z\in \mathcal{Q}_{\F^S}}\Big\{ \EG{\frac{d\mathbb{Q}_{\G}}{d\mathbb{P}_{\G}}Z H^S}+\EG{\frac{d\mathbb{Q}_{\G}}{d\mathbb{P}_{\G}}Z H}-c_\G(Z)\Big\}\\
&=\esssup_{Z\in \mathcal{Q}_{\F^S}} \Big\{\EG{\EGS{\frac{d\mathbb{Q}_{\G}}{d\mathbb{P}_{\G}}ZH^S}}+\EG{\frac{d\mathbb{Q}_{\G}}{d\mathbb{P}_{\G}}Z H}-c_\G(Z)\Big\}
\\
&=\esssup_{Z\in \mathcal{Q}_{\F^S}} \Big\{\EG{\frac{d\mathbb{Q}_{\G}}{d\mathbb{P}_{\G}} H^S\EGS{Z}}+\EG{\frac{d\mathbb{Q}_{\G}}{d\mathbb{P}_{\G}}Z H}-c_\G(Z)\Big\}
\\
&=\esssup_{Z\in \mathcal{Q}_{\F^S}} \Big\{\EG{\frac{d\mathbb{Q}_{\G}}{d\mathbb{P}_{\G}} H^S}+\EG{\frac{d\mathbb{Q}_{\G}}{d\mathbb{P}_{\G}}Z H}-c_\G(Z)\Big\}
\\
&=\EG{\frac{d\mathbb{Q}_{\G}}{d\mathbb{P}_{\G}} H^S}+\esssup_{Z\in \mathcal{Q}_{\F^S}} \Big\{\EG{\frac{d\mathbb{Q}_{\G}}{d\mathbb{P}_{\G}}Z H}-c_\G(Z)\Big\}
=\EQG{H^S}+\Pi_{\G}(H),
\end{align*}
where we have used in the fourth equation that
$d\mathbb{Q}_{\G}/d\mathbb{P}_{\G}$ and $H^S,$ by assumption,
are $\F^S$-measurable. In the fifth equation we have used that
$\EGS{Z}=1$. This proves (ii)$\Rightarrow$(iii). The direction
(iii)$\Rightarrow$(i) is clear using that by normalization
$\Pi_{\G}(0)=0$.

Finally let us show that if $\Pi_{\G}$ is additionally assumed
to be monotone, then market-consistency is equivalent to any
payoff being evaluated between its super- and sub-replication
price. Clearly, if any payoff is evaluated between its super-
and sub-replication price then for any $H^S\in L^\infty(\F^S)$
we must have that $\Pi_\G(H^S)=\EQG{H^S},$ as the financial
market is assumed to be complete. Hence, by the direction
(i)$\Rightarrow$(iii) shown above, indeed $\Pi_\G$ is
market-consistent. On the other hand if $\Pi_\G$ is
market-consistent and monotone then by standard duality results
the penalty function in (ii) must have a domain in $
\mathcal{Q}^+_{\F^S}.$ Also note that the set defined by
$M:=\dfrac{d\mathbb{Q}_{\G}}{d\mathbb{P}_{\G}}\mathcal{Q}^+_{\F^S}=\{\dfrac{d\mathbb{Q}_{\G}}{d\mathbb{P}_{\G}}Z|Z\in
\mathcal{Q}^+_{\F^S}\}$ is equal to the set of all local
martingale measures $\mathcal{M}_\G$. (Actually for our proof
we only need that $M\subset \mathcal{M}_\G$.) This yields that
$$\Pi_{\G}(H)= \esssup_{Z\in \mathcal{Q}^+_{\F^S}}
\{\EQG{ZH}-c_\G(Z)\}\leq \esssup_{Z\in \mathcal{Q}^+_{\F^S}}
\EQG{ZH}= \esssup_{\bar{\mathbb{P}}_{\G}\in
\mathcal{M}_\G}\Ebarr{H},$$ where we have used in the first
equality that (ii) holds as $\Pi_\G$ is market-consistent. In
the inequality we applied that $c_{\G}\geq 0.$ In the last
equality we used that $M= \mathcal{M}_{\G}.$ Hence, indeed
$\Pi_{\G}(H)$ is smaller than the super-replication price of
$H.$ To show that $\Pi_{\G}(H)$ is greater than the
sub-replication price, note that as $c$ is a penalty function
we must have $\essinf_{Z\in \mathcal{Q}^+_{\F^S}} c_\G(Z)=0.$
Now clearly, $\mathcal{Q}^+_{\F^S}$ satisfies the concatenation
property from Lemma A.3. Thus, there exists a sequence $Z^n\in
\mathcal{Q}^+_{\F^S}$ such that a.s. $\lim_n c_{\G}(Z^n)=0.$
This entails that
\begin{align*}\Pi_{\G}(H)&= \esssup_{Z\in \mathcal{Q}^+_{\F^S}}
\{\EQG{ZH}-c_\G(Z)\} \\
&\geq\lim_n\{\EQG{Z^nH}-c_\G(Z^n)\}\\
&\geq
\liminf_n\EQG{Z^nH}+\liminf_n -c_\G(Z^n)
=\liminf_n\EQG{Z^nH}\geq \essinf_{\bar{\mathbb{P}}_{\G}\in
\mathcal{M}_\G}\Ebarr{H}.
\end{align*}
Hence, $\Pi_{\G}(H)$ is greater than the sub-replication price.
The proposition is proved.
{\unskip\nobreak\hfill$\Box$\par\addvspace{\medskipamount}}

Theorems \ref{main} and \ref{main2} may be seen as versions of
the Radon-Nikodyn theorem with a non-linear part $\Pi_{\F^S}$
and without assumptions like  monotonicity or continuity. We
will need the following Lemma. Its proof is straight-forward
doing an induction over $r.$
\begin{lemma}
\label{additive} Suppose that $\Pi_\G$ satisfies
market-consistency and the local property. Then for disjoint
sets $C_{1},\ldots,C_{r}\in\F^S$ and payoffs $H_1,\ldots,H_r\in
L^\infty(\F),$ we have
$$ \Pi_{\G}(H_1 I_{C_{1}}+\ldots+H_r I_{C_{r}})=\Pi_{\G}(H_1 I_{C_{1}})+\ldots+\Pi_{\G}(H_r I_{C_{r}}).$$
\end{lemma}

The next lemma will also be useful.
\begin{lemma}
\label{lemmalocal} A $\bar{\G}$-conditional evaluation
satisfies $\Pi_{\bar{\G}}(I_A H)=I_A\Pi_{\bar{\G}}( H)$ for
every $A\in\bar{\G}.$
\end{lemma}
\begin{proof}
It is
\begin{align*}
\Pi_{\bar{\G}}(I_A H)=\Pi_{\bar{\G}}(I_A H+I_{A^c}0)=I_A\Pi_{\bar{\G}}( H)+I_{A^c}\Pi_{\bar{\G}}(0)=I_A\Pi_{\bar{\G}}( H),
\end{align*}
by normalization and the local property of $\Pi_{\bar{\G}}.$
\end{proof}
 For
the proof of Theorem \ref{main} we will also need the following
lemma:
\begin{lemma}
\label{concatenation} In the setting of Theorem \ref{main}, the
set $M'$ defined by (\ref{M}) has the concatenation property in
the sense that $Z_1,Z_2\in M'$ implies that for any $A\in \F^S$
we have that $Z_1I_A+Z_2I_{A^c}\in M'.$ In particular,
$I_AM'+I_{A^c}M':=\{Z_1I_A+Z_2I_{A^c}|Z_1\in M', Z_2\in
M'\}=M'.$
\end{lemma}
\begin{proof}
For $Z_1,Z_2\in M'$ we have
\begin{align}
\label{pstarconc}
\Pi^*_{\G}\Big(\frac{d\mathbb{Q}_{\G}}{d\mathbb{P}_{\G}}&\big(Z_1I_A+Z_2I_{A^c}\big)\Big)\nonumber\\
&=\esssup_{H\in L^\infty(\F)} \Big\{\EG{\frac{d\mathbb{Q}_{\G}}{d\mathbb{P}_{\G}}(Z_1I_A+Z_2I_{A^c})H}-\Pi_{\G}(H)\Big\}\nonumber\\
&=\esssup_{H\in L^\infty(\F)} \Big\{\EG{\frac{d\mathbb{Q}_{\G}}{d\mathbb{P}_{\G}}Z_1(I_AH)}+\EG{\frac{d\mathbb{Q}_{\G}}{d\mathbb{P}_{\G}}Z_2(I_{A^c}H)}-\Pi_{\G}(HI_A+HI_{A^c})\Big\}
\nonumber\\
&=\esssup_{HI_A,HI_{A^c}\in L^\infty(\F)} \Big\{\EG{\frac{d\mathbb{Q}_{\G}}{d\mathbb{P}_{\G}}Z_1(I_AH)}-\Pi_{\G}(HI_A)
\nonumber\\&\hspace{3.5cm}
+\EG{\frac{d\mathbb{Q}_{\G}}{d\mathbb{P}_{\G}}Z_2(I_{A^c}H)}-\Pi_{\G}(HI_{A^c})\Big\}\nonumber\\
&=
\esssup_{HI_A\in L^\infty(\F)} \Big\{\EG{\frac{d\mathbb{Q}_{\G}}{d\mathbb{P}_{\G}}Z_1 (HI_A)}-\Pi_{\G}(HI_A)\Big\}
\nonumber\\&\hspace{0.5cm}+\esssup_{HI_{A^c}\in L^\infty(\F)} \Big\{\EG{\frac{d\mathbb{Q}_{\G}}{d\mathbb{P}_{\G}}Z_2(HI_{A^c})}-\Pi_{\G}(HI_{A^c})\Big\},
\end{align}
where the first equation holds by the definition of
 $\Pi^*.$ The third equation holds because of Lemma \ref{additive} for $r=2$.
(\ref{pstarconc}) yields that
\begin{align*}
\Pi^*_{\G}\Big(\frac{d\mathbb{Q}_{\G}}{d\mathbb{P}_{\G}}\big(Z_1I_A+Z_2I_{A^c}\big)\Big)
&\leq \esssup_{H\in L^\infty(\F)} \Big\{\EG{\frac{d\mathbb{Q}_{\G}}{d\mathbb{P}_{\G}}Z_1 H}-\Pi_{\G}(H)\Big\}\nonumber\\
&\hspace{0.5cm}+\esssup_{H\in L^\infty(\F)} \Big\{\EG{\frac{d\mathbb{Q}_{\G}}{d\mathbb{P}_{\G}}Z_2H}-\Pi_{\G}(H)\Big\}\nonumber\\
&=\Pi^*_{\G}(\frac{d\mathbb{Q}_{\G}}{d\mathbb{P}_{\G}}Z_1)+\Pi^*_{\G}(\frac{d\mathbb{Q}_{\G}}{d\mathbb{P}_{\G}}Z_2)=0+0=0,
\end{align*}
where in the last line we have used that
$\Pi^*_{\G}(\frac{d\mathbb{Q}_{\G}}{d\mathbb{P}_{\G}}Z_i)=0.$
The reason for this is that $\Pi^*$ is zero on $M$ and infinity
else. Therefore, the fact that $Z_i\in M'$ implies by (\ref{M})
that $\frac{d\mathbb{Q}_{\G}}{d\mathbb{P}_{\G}}Z_i\in \frac{d\mathbb{Q}_{\G}}{d\mathbb{P}_{\G}}M'=M$ for $i=1,2.$\\

Since $\Pi^*_{\G}$ is only takes the values zero and infinity
we must have that
$\Pi^*_{\G}(\frac{d\mathbb{Q}_{\G}}{d\mathbb{P}_{\G}}(Z_1I_A+Z_2I_{A^c}))=0.$
Thus, we can conclude that indeed
$\frac{d\mathbb{Q}_{\G}}{d\mathbb{P}_{\G}}\big(Z_1I_A+Z_2I_{A^c}\big)\in
M.$ Therefore,
$$Z_1I_A+Z_2I_{A^c}\in \Big(\frac{d\mathbb{Q}_{\G}}{d\mathbb{P}_{\G}}\Big)^{-1}M=M'.$$
\end{proof}
\textit{Proof of Theorem \ref{main}.} (ii)$\Rightarrow$(i): It
is
$$
\Pi_{\G}(H^S+H)=\EQG{\Pi_{\F^S}(H^S+H)}
=\EQG{H^S}+\EQG{\Pi_{\F^S}(H)}=\EQG{H^S}+\Pi_{\G}(H),
$$
where we have used $\F^S$-conditional cash invariance in the
second equation. This shows market-consistency.
Moreover,
\begin{align*}
\Pi_{\G}(I_AH_1+I_{A^c} H_2)&=\EQG{\Pi_{\F^S}(I_AH_1+I_{A^c} H_2)}\\
&=\EQG{I_A\Pi_{\F^S}(H_1)+I_{A^c}\Pi_{\F^S}(H_2)}\\
&=\EQG{\Pi_{\F^S}(I_AH_1)+\Pi_{\F^S}(I_{A^c} H_2)}
=\Pi_{\G}(I_AH_1)+\Pi_{\G}(I_{A^c}H_2),
\end{align*}
where we used the $\F^S$-local property in the second and Lemma
\ref{lemmalocal} in the third equation.

(i)$\Rightarrow$(ii): By Proposition \ref{prop}
(iii)$\Rightarrow$(ii) and positive homogeneity
market-consistency imply that
\begin{align}
\label{first}
\Pi_{\G}(H)&=\esssup_{Z\in M'} \EG{\frac{d\mathbb{Q}_{\G}}{d\mathbb{P}_{\G}}Z H}\nonumber\\
&=\esssup_{Z\in M'} \EG{\EGS{\frac{d\mathbb{Q}_{\G}}{d\mathbb{P}_{\G}}Z H}}=\esssup_{Z\in M'} \EG{\frac{d\mathbb{Q}_{\G}}{d\mathbb{P}_{\G}}\EGS{Z H}},
\end{align}
where we used in the third equation that
$\frac{d\mathbb{Q}_{\G}}{d\mathbb{P}_{\G}}$ is
$\F^S$-measurable. Define
$$\Pi_{\F^S}(H):=\esssup_{Z\in M'}\EGS{Z H}.$$
Clearly, $\Pi_{\F^S}$ is normal, $\F^S$-convex, $\F^S$-cash
invariant and $\F^S$-positively homogeneous. The $\F^S$-local
property is satisfied because for $A\in \F^S$ we have
\begin{align*}
\Pi_{\F^S}(I_AH_1+I_{A^c} H_2)&=\esssup_{Z\in M'}\EGS{ZI_A H_1+ZI_{A^c}H_2}\\
&=\esssup_{Z\in M'}\EGS{ZI_A H_1}+\EGS{ZI_{A^c}H_2}\\
&=\esssup_{Z_1\in M',Z_2\in M'}\EGS{Z_1I_A H_1}+\EGS{Z_2I_{A^c}H_2}\\
&=\esssup_{Z_1\in M'}\EGS{Z_1I_A H_1}+\esssup_{Z_2\in M'}\EGS{Z_2I_{A^c}H_2}
\\&=I_A\esssup_{Z_1\in M'}\EGS{Z_1 H_1}+I_{A^c}\esssup_{Z_2\in M'}\EGS{Z_2H_2}\\
&=I_A\Pi_{\F^S}(H_1)+I_{A^c}\Pi_{\F^S}(H_2),
\end{align*}
where we used in the third equation that Lemma
\ref{concatenation} implies that $M'=\{Z_1I_A+Z_2I_{A^c}|Z_1\in
M', Z_2\in M'\}.$ Hence, indeed $\Pi_{\F^S}$ is an
$\F^S$-conditional evaluation. Finally, let us prove that
$$\Pi_{\G}(H)=\EQG{\Pi_{\F^S}(H)}.$$
Notice that if we could show that \be \label{second}
\esssup_{Z\in M'}
\EG{\frac{d\mathbb{Q}_{\G}}{d\mathbb{P}_{\G}}\EGS{Z H}}
=\EG{\frac{d\mathbb{Q}_{\G}}{d\mathbb{P}_{\G}}\esssup_{Z\in
M'}\EGS{Z H}} \ee then we are done, since the left-hand side of
(\ref{second}) is equal to $\Pi_{\G}(H)$ by (\ref{first}),
while the righthand-side is equal to $\EQG{\Pi_{\F^S}(H)}$ by
the definition of $\Pi_{\F^S}.$ So let us show (\ref{second}).
Clearly,
$$\esssup_{Z\in M'} \EG{\frac{d\mathbb{Q}_{\G}}{d\mathbb{P}_{\G}}\EGS{Z H}}
\leq \EG{\frac{d\mathbb{Q}_{\G}}{d\mathbb{P}_{\G}}\esssup_{Z\in M'}\EGS{Z
H}}.$$ Let us prove `$\geq$'. It is well known, see also
Theorem \ref{Schied} and Lemma \ref{appendix}, that the
concatenation property implies that there exists a sequence
$Z_n\in M'$ with $\EGS{Z_1 H}\leq \EGS{Z_2 H}\leq\ldots$ such
that $\lim_n\EGS{Z_n H}=\esssup_{Z\in M'}\EGS{Z H}.$ Therefore,
by the monotone convergence theorem
\begin{align*}
\EG{\frac{d\mathbb{Q}_{\G}}{d\mathbb{P}_{\G}}\esssup_{Z\in M'}\EGS{Z H}}&=
\lim_n\EG{\frac{d\mathbb{Q}_{\G}}{d\mathbb{P}_{\G}}\EGS{Z_n H}}\\
&\leq \esssup_{Z\in M'} \EG{\frac{d\mathbb{Q}_{\G}}{d\mathbb{P}_{\G}}\EGS{Z H}}.
\end{align*}
This shows (\ref{second}). This proves Theorem \ref{main}.
{\unskip\nobreak\hfill$\Box$\par\addvspace{\medskipamount}}

For the proof of Theorem \ref{main2} we will need the following
Corollary of Lemma \ref{additive}:
\begin{corollary}
\label{infiniteadditive} Suppose that $\Pi_\G$ is continuous,
market-consistent and satisfies the local property. Then for
disjoint sets $C_{1},C_2,\ldots\in\F^S$ and a payoff $H\in
L^\infty(\F)$ we have
$$ \Pi_{\G}(H \sum_{i=1}^\infty I_{C_{i}})=\sum_{i=1}^\infty\Pi_{\G}(H I_{C_i}).$$
In particular, if $\Pi_\G$ is additionally assumed to be
monotone (or $p$-norm bounded) then for every $H\in
L^\infty(\F)$ the mapping $C\to \Pi_{\G}(H I_{C})$ is a
real-valued, (signed) measure on $\F^S$.
\end{corollary}
\textit{Proof of Theorem \ref{main2}.} (ii)$\Rightarrow$(i):
Continuity, convexity, and monotonicity in the case that
$\Pi_{\F^S}$ is monotone, are straight-forward. The other
properties are seen analogously as in the proof of Theorem
\ref{main}.

(i)$\Rightarrow$(ii): By assumption
$\F^S=\sigma(A_1,A_2,\ldots)$ and we will assume without loss
of generality that $A_i\ne A_j$ if $i\ne j.$ For
$n\in\mathbb{N}$ we define the finite filtration $\F^n$ as the
smallest $\sigma$-algebra containing the the events
$A_1,\ldots,A_n.$ Now let us define the partitions
corresponding to $\F^n$ recursively in a standard way. For
$n=1$, $\F^1$ is generated by the partition given by
$B^1_1:=A_1$, $B^1_2:=A^c_1.$ Moreover, $\F^{n+1}$ is generated
by the partition \beas  B^{n+1}_i:=B^{n}_i\cap A_{n+1} \mbox{
and } B^{n+1}_{2^n+i}:=B^{n}_i\cap
A^c_{n+1}\,\,\,\,\,i=1,\ldots,2^n. \eeas
Of course, $\F^1\subset \F^2\subset \ldots \subset \F^S$.

Set $q^n_k=\mathbb{Q}_\G(B^n_k).$ Note that $q^n_k$ are
$\G$-measurable random variables summing up to one for a.s. all
fixed $\omega.$ Define \be \label{defcond} \Pi_{\F^n}(H):=
\sum_{k=1}^{2^n}\frac{I_{B^n_k}}{q^n_k}\Pi_{\G}(HI_{B^n_k}),
\ee where we we set $0/0=0$. Note that if
$q^n_k=\mathbb{Q}_\G[B^n_k]=0$ on a non-zero set $C\in \G$,
then $\mathbb{P}_\G[B^n_k]=0$ on $C$ as well (as
$\mathbb{Q}_{\G}$ is equivalent to $\mathbb{P}_{\G}$)
 and therefore, $\Pi_\G(HI_{B^n_k})=0$ on $C$, too. In particular, $\Pi_{\F^n}$ is well defined.

Next, every set $A\in \F^n$ can be written as
$A=B^n_{k_1}\cup\ldots\cup B^n_{k_r},$ for a
$r\in\{1,\ldots,2^n\}$ and $1\leq k_1<k_2<\ldots<k_r\leq 2^n.$
As $\Pi_\G(0)=0$, it is straight-forward to check using
Definition (\ref{defcond}) that
\begin{align}
\label{local}
\Pi_{\F^n}(I_A H)
=I_A\Pi_{\F^n}(H).
\end{align}
Furthermore, by Lemma \ref{additive} for every $H\in
L^\infty(\F)$
\begin{align*}
\EQG{\Pi_{\F^n}( H)}&=\sum_{k=1}^{2^n}\frac{\EQG{I_{B^n_{k}}}}{q^n_{k}}\Pi_{\G}( HI_{B^n_{k}})
=\Pi_{\G}(H).
\end{align*}
We will need the following lemma which is a version of
Proposition A.12 in \citet{follmer2004stochastic}.
\begin{lemma}
\label{A.12} Suppose that $\hat{\mathbb{P}}$ and
$\bar{\mathbb{P}}$ share the same zero sets, and that
$\F_0\subset \F$. Then for any bounded $\F$-measurable $H$
$$\Ehatf{H}=\frac{1}{\Ebarf{d\hat{\mathbb{P}}/d\bar{\mathbb{P}}}}\Ebarf{\frac{d\hat{\mathbb{P}}}{d\bar{\mathbb{P}}} H}\quad \bar{\mathbb{P}}\mbox{ and }
\hat{\mathbb{P}}\mbox{ a.s.}$$
\end{lemma}
Next we will show the following Lemma:
\begin{lemma}
Under our assumptions, for every $H\in L^\infty(\F)$ the
process $M_n=\EQG{\Pi_{\F^n}(H)}$ is a uniformly integrable
martingale.
\end{lemma}
\begin{proof}
It may be seen from standard arguments that $M_n=\Pi_{\F^n}(H)$
is a martingale. Let us see that $M_n$ is uniformly integrable
under the measure $\mathbb{Q}_\G$. First of all note that in
the case that $\Pi_\G$ is monotone we have that
$$\Pi_{\F^n}(H) \leq
\sum_{k=1}^{2^n}\frac{I_{B^n_k}}{q^n_k}\Pi_{\G}(||H||_\infty I_{B^n_k})
=\sum_{k=1}^{2^n}\frac{I_{B^n_k}}{q^n_k}||H||_\infty q^n_k= ||H||_\infty,
$$
where we used market-consistency in the first equation.
Similarly, it is seen that $\Pi_{\F^n}(H)\geq -||H||_\infty.$
Hence, if $\Pi_{\F^n}$ is monotone then for fixed $H,$
$\Pi_{\F^n}(H)$ is bounded uniformly. In particular, $M_n$ is
uniformly integrable.

In the case that $\Pi_\G$ is $p$-norm bounded notice that,
\begin{align*}
\big|\Pi_{\F^n}(H) \big|&\leq
\l\Big|\sum_{k=1}^{2^n}\frac{I_{B^n_k}}{q^n_k}\int (|H|+|H|^p) I_{B^n_k}d\bar{\mathbb{P}}_\G\Big|\\
&
\leq \l(||H||_\infty+ ||H||^{p}_\infty) \Big|\sum_{k=1}^{2^n}I_{B^n_k}\frac{\bar{\mathbb{P}}_\G(B^n_k)}{\mathbb{Q}_\G(B^n_k)} \Big|=
\l(||H||_\infty+ ||H||^{p}_\infty)\EQG{\frac{d\bar{\mathbb{P}}_\G}{d\mathbb{Q}_\G}|\F^n}.
\end{align*}
Since the last term is uniformly integrable, $\Pi_{\F^n}(H)$ is
uniformly integrable as well.
\end{proof}
Hence, if $\Pi_\G$ is monotone or $p$-norm bounded we may
conclude by the martingale convergence theorem that
$M_n=\Pi_{\F^n}(H)$ converges a.s. and in $L^1(\mathbb{Q}_\G)$
to a random variable $M_\infty.$ Set
$$\Pi_{\F^S}(H):=M_\infty=\lim_n\Pi_{\F^n}(H).$$
Now $\Pi_{\G}(H)=\EQG{\Pi_{\F^S}(H)}$ will follow from the
following lemma:
\begin{lemma}
\label{lemmaone} For all $A\in \F^S$ and $H\in L^\infty(\F),$
$\Pi_{\F^S}(H)$ satisfies the \emph{characteristic equation}:
\be \label{characteristic} \Pi_\G(I_A H)=\EQG{I_A
\Pi_{\F^S}(H)}. \ee In particular,
$\Pi_\G(H)=\EQG{\Pi_{\F^S}(H)}.$ Furthermore, for every $H\in
L^\infty(\F^S)$, $\Pi_{\F^S}(H)$ is the unique a.s.
$\mathbb{Q}_\G$-integrable random variable satisfying
(\ref{characteristic}).
\end{lemma}
\begin{proof}
Notice that for every $A\in \F^{n_0}\subset
\F^{n_0+1}\subset\ldots $ by (\ref{local})
\begin{align}
\label{fnsets}
\Pi_{\G}(I_A H)=\lim_n \EQG{\Pi_{\F^n}(I_A H)}=
\lim_n \EQG{I_A\Pi_{\F^n}(H)}
=\EQG{I_A \Pi_{\F^S}(H)},
\end{align}
where we have used in the third equality that $\Pi_{\F^n}(H)$
converges to $\Pi_{\F^S}(H)$ in $L^1(\mathbb{Q}_\G).$ Now by
Corollary \ref{infiniteadditive} the left hand-side and the
right hand-side are signed measures (measuring sets $A\in
\F^S$). As by (\ref{fnsets}) they both agree on $\bigcup_n
\F^n,$ which is closed under intersection and generates $\F^S$,
they must agree on the entire filtration $\F^S.$ Uniqueness
follows from equation (\ref{characteristic}) using standard
arguments.
\end{proof}

Now all what is left to prove is the following lemma:
\begin{lemma}
\label{lemmatwo} $\Pi_{\F^S}$ is a continuous,
$\F^S$-conditional evaluation which is monotone if $\Pi_\G$ is
monotone.
\end{lemma}
\begin{proof}
Clearly  for every $n$ $\Pi_{\F^n}$ is $\F^n$- and hence also
$\F^S$-measurable. This entails that its limit, $\Pi_{\F^S},$
is $\F^S$-measurable. Furthermore, by construction,
$\Pi_{\F^S}$ is normalized as $\Pi_{\F^n}$ are. If $\Pi_{\G}$
is monotone then the $\Pi_{\F^n}$ are monotone as well which
implies that $\Pi_{\F^S}$ is monotone.

Next, let us check that $\Pi_{\F^S}$ satisfies the $\F^S$-local
property. It is necessary and sufficient that for every $A\in
\F^S$ and $H\in L^\infty(\F)$
$$\Pi_{\F^S}(HI_A)=I_A \Pi_{\F^S}(H).$$
We will prove the equality by showing that the right-hand side
satisfies, the characteristic equation (\ref{characteristic})
for the left hand-side. So let $A' \in \F^S$. It is
$$\Pi_\G(I_{A'} (I_A H))=\Pi_\G(I_{A'\cap A} H)=\EQG{I_{A'\cap A} \Pi_{\F^S}(H)}=\EQG{I_{A'} \Big(I_A\Pi_{\F^S}(H)\Big)},$$
where we have used (\ref{characteristic}) in the second
equation. This shows that $I_A \Pi_{\F^S}(H)$ satisfies the
characteristic equation of $I_A H$ and hence by the uniqueness
stated in Lemma \ref{lemmaone} indeed $\Pi_{\F^S}(HI_A)=I_A
\Pi_{\F^S}(H).$

To see $\F^S$-cash invariance of $\Pi_{\F^S}$ assume that for a
$m\in L^\infty(\F^S)$ and $H\in L^\infty(\F)$ we have that the
set $A=\{\Pi_{\F^S}(H+m)\stackrel{(<)}{>} \Pi_{\F^S}(H)+m\}$
has positive measure under $\mathbb{P}_\G$. Then $A$ has also
positive measure under $\mathbb{Q}_\G$ and
\begin{align*}
\Pi_{\G}(HI_A +mI_A)&=\EQG{\Pi_{\F^S}((H+m)I_A)}=
\EQG{\Pi_{\F^S}(H+m)I_A}\\
&\stackrel{(<)}{>}\EQG{(\Pi_{\F^S}(H )+m) I_A }\\
&=\EQG{\Pi_{\F^S}(HI_A)}+\EQG{mI_A} =\Pi_{\G}(HI_A)+\EQG{mI_A},
\end{align*}
where we have used in the second equation the $\F^S$-local
property of $\Pi_{\F^S},$ which we have proved above. This is a
contradiction to the market-consistency of $\Pi_\G$. Thus,
indeed $\Pi_{\F^S}$ is $\F^S$-cash invariant.

Next, for $H\in L^\infty(\F),$ and $\l\in\mathbb{R}$ with
$0\leq \l\leq 1$
\begin{align}
\label{fnconv}
\Pi_{\F^n}(\l H_1+(1-\l)H_2)&=
\sum_{k=1}^{2^n}\frac{I_{B^n_k}}{q^n_k}\Pi_{\G}(\l I_{B^n_k} H_1+(1-\l)I_{B^n_k} H_2)\nonumber\\
&\leq
\sum_{k=1}^{2^n}\frac{I_{B^n_k}}{q^n_k}\Big(\l \Pi_{\G}(I_{B^n_k} H_1)+(1-\l)\Pi_{\G}(I_{B^n_k} H_2)\Big)\nonumber\\
&=\l\Pi_{\F^n}(H_1)+(1-\l)\Pi_{\F^n}H_2),
\end{align}
where we have used the convexity of $\Pi_\G$ in the inequality.
In particular, $\Pi_{\F^S}$ is convex as limit of convex
functionals. (Note that `convexity' is weaker than
`$\F^S$-convexity' since for `convexity' $\l\in[0,1]$ is
assumed to be deterministic.) Before we move on to show the
$\F^S$-convexity of $\Pi_{\F^S}$ let us prove continuity.

In the case that $\Pi_{\G}$ is monotone, $\Pi_{\F^S}$ is
monotone as well and the continuity of $\Pi_{\F^S}$ follows
from the characteristic equation using standard arguments.

On the other hand if $\Pi_{\G}$ is $p$-norm bounded we get for
$H\in L^\infty(\F)$
\begin{align*}
\big|\Pi_{\F^n}(H) \big|&\leq
\l\Big|\sum_{k=1}^{2^n}\frac{I_{B^n_k}}{q^n_k}\int (|H|+|H|^p) I_{B^n_k}d\bar{\mathbb{P}}_\G\Big|\\
&=\l\Big|\sum_{k=1}^{2^n}\frac{I_{B^n_k}}{q^n_k}\int (|H|+|H|^p) I_{B^n_k}
\frac{d\bar{\mathbb{P}}_\G}{d\mathbb{Q}_\G}d\mathbb{Q}_\G\Big|=\l\EQG{(|H|+|H|^p)\frac{d\bar{\mathbb{P}}_\G}{d\mathbb{Q}_\G}\Big|\F^n}.
\end{align*}
Therefore,
\begin{align}
\label{pseudop}
\big|\Pi_{\F^S}(H) |&=\lim_n\big|\Pi_{\F^n}(H) \big|\nonumber\\
&\leq \l\EQG{(|H|+|H|^p) \frac{d\bar{\mathbb{P}}_\G}{d\mathbb{Q}_\G}\Big|\F^S}
\nonumber\\&=\l\EQG{\frac{d\bar{\mathbb{P}}_\G}{d\mathbb{Q}_\G}|\F^S}
\frac{\EQG{(|H|+|H|^p) \frac{d\bar{\mathbb{P}}_\G}{d\mathbb{Q}_\G}\Big|\F^S}}{\EQG{\frac{d\bar{\mathbb{P}}_\G}{d\mathbb{Q}_\G}|\F^S}}
\nonumber\\
&=
\l
\EQG{\frac{d\bar{\mathbb{P}}_\G}{d\mathbb{Q}_\G}|\F^S}
\int (|H|+|H|^p) d\bar{\mathbb{P}}_{\F^S}.
\end{align}
The third equation holds by Lemma \ref{A.12}. 
Now by (\ref{pseudop}) we can
extend $\Pi_{\F^S}$ to $L^p(\Omega,\F,\bar{\mathbb{P}})$
by setting 
$$\tilde{\Pi}_{\F^S}(H)=\limsup_{N\to\infty}\limsup_{m\to\infty}
\Pi_{\F^S} (-m\vee H \wedge N).$$ Note that we have that
$\tilde{\Pi}_{\F^S}$ is convex (as the $\limsup$ of convex
functionals), and agrees with $\Pi_{\F^S}$ on $
L^\infty(\F).$
Define 
$L^p(\Omega,\F,\bar{\mathbb{P}}_{\F^S})$ as all random variables $H$ such that
$\Big(\int |H(\omega')|^p(\omega')\bar{\mathbb{P}}_{\F^S})(d\omega')\Big)^{1/p}:=\Big(\EFS{|H|^p \bar{\xi}}\Big)^{1/p}<\infty,$ 
where $\bar{\xi}$ is the conditional density corresponding to $\bar{\mathbb{P}}_{\F^S}.$
By (\ref{pseudop}), for every $H$ we have that $\tilde{\Pi}_{\F^S}(H)$ is real-valued and uniformly
and bounded in any $L^p(\Omega,\F,\bar{\mathbb{P}}_{\F^S})$
environment around $H.$
It follows then from standard arguments for convex functionals, see for instance Theorem 2.2.9 in
\citet{zalinescu2002convex}, that convergence of $H_n$ to $H$ in 
$L^p(\Omega,\F,\bar{\mathbb{P}}_{\F^S})$ implies that 
$\tilde{\Pi}_{\F^S}(H_n)$ converges to $\tilde{\Pi}_{\F^S}(H)$.
Since $\tilde{\Pi}_{\F^S}$ and
$\Pi_{\F^S}$ agree on $L^\infty(\F)$ we may concluded that indeed 
$\Pi_{\F^S}$ is continuous with respect to bounded a.s. convergence.

Finally let us show that $\Pi_{\F^S}$ is $\F^S$-convex. First
of all, let $\l^{\F^n}\in L^\infty(\F^n)$ with $0\leq
\l^{\F^n}\leq 1.$ Then there exists disjoint sets
$A^n_{1},\ldots, A^n_{r}\in \F^n$ and constants
$\l^n_{1},\ldots, \l^n_{r}\in[0,1]$ with
$\l^{\F^n}=\sum_{j=1}^r \l^n_{i}I_{A^n_{i}}.$ By adding an
additional set with an additional constant equal to zero if
necessary, we may assume without loss of generality that
$\Omega=\bigcup_{i=1}^r A^n_{i}.$ By Lemma \ref{additive} for
sets $C_1,\ldots,C_r\in \F^S$
\begin{align}
\label{fnlocal}
\Pi_{\F^n}(\sum_{i=1}^r I_{C_i} H_i)&=
\sum_{k=1}^{2^n}\frac{I_{B^n_k}}{q^n_k}\Big(\Pi_{\G}(\sum_{i=1}^r I_{B^n_k} I_{C_i}H_i)\Big)\nonumber\\
&=
\sum_{i=1}^r \sum_{k=1}^{2^n}\frac{I_{B^n_k}}{q^n_k}\Big(\Pi_{\G}(I_{B^n_k} I_{C_i}H_i)\Big)
=\sum_{i=1}^r \Pi_{\F^n}(I_{C_i} H_i).
\end{align}
Therefore,
\begin{align*}
\Pi_{\F^n}(\l^{\F^n} H_1+(1-\l^{\F^n})H_2)&=\Pi_{\F^n}\Big(\sum_{i=1}^rI_{A^n_i}( \l^n_iH_1+(1-\l^{n}_i)H_2)\Big)\\
&=\sum_{i=1}^r\Pi_{\F^n}( \l^n_iI_{A^n_i}H_1+(1-\l^{n}_i)I_{A^n_i}H_2)
\\&\leq \sum_{i=1}^r \l^n_i \Pi_{\F^n}(I_{A^n_i}H_1)+(1-\l^{n}_i)\Pi_{\F^n}(I_{A^n_i}H_2)\\
&=\sum_{i=1}^r \l^n_i I_{A^n_i}\Pi_{\F^n}(H_1)+\sum_{i=1}^r(1-\l^{n}_i)I_{A^n_i}\Pi_{\F^n}(H_2)\\
&=\l^{\F^n} \Pi_{\F^n}(H_1)+(1-\l^{\F^n})\Pi_{\F^n}(H_2),
\end{align*}
where we have used (\ref{fnlocal}) in the second equation,
(\ref{fnconv}) in the inequality, and the $\F^n$-local property
for $\Pi_{\F^n},$ proved in (\ref{local}), in the third
equation. Hence, indeed $\Pi_{\F^n}$ is $\F^n$-convex.

Next, if $\l^{\F^m}\in L^\infty(\F^m)$ with $0\leq
\l^{\F^m}\leq 1,$ and $m\in\mathbb{N}$ then clearly
$\l^{\F^m}\in \F^n$ for every $n\geq m.$ This entails
\begin{align}
\label{fsmconv}
\Pi_{\F^S}(\l^{\F^m} H_1+(1-\l^{\F^m})H_2)&=\lim_n \Pi_{\F^n}(\l^{\F^m} (H_1+(1-\l^{\F^m})H_2)\nonumber\\
&\leq
\lim_n \l^{\F^m} \Pi_{\F^n}(H_1)+(1-\l^{\F^m})\Pi_{\F^m}(H_2)\nonumber\\
&=\l^{\F^m} \Pi_{\F^S}(H_1)+(1-\l^{\F^m})
\Pi_{\F^m}(H_2).
\end{align}
Finally, to see that $\Pi_{\F^S}$ is $\F^S$-convex let
$\l^{\F^S}\in L^\infty(\F^S)$ with $0\leq \l^{\F^S}\leq 1.$
Then also $0\leq \EFn{\l^{\F^S}}\leq 1$. Furthermore, by the
martingale convergence theorem $\EFn{\l^S}$ converges to
$\EFS{\l^S}=\l^S$ a.s. Hence, by the continuity of $\Pi_{\F^S}$
we obtain that
\begin{align*}
\Pi_{\F^S}(\l^{\F^S} H_1+(1-\l^{\F^S})H_2)&=\lim_n
\Pi_{\F^S}\Big(\EFn{\l^{\F^S}} H_1+\big(1-\EFn{\l^{\F^S}}\big)H_2\Big)\\
&\leq \lim_n
\EFn{\l^{\F^S}} \Pi_{\F^S}(H_1)+\big(1-\EFn{\l^{\F^S}}\big)
\Pi_{\F^S}(H_2)\\
&=
\l^{\F^S} \Pi_{\F^S}(H_1)+(1-\l^{\F^S})
\Pi_{\F^S}(H_2),
\end{align*}
where we have used (\ref{fsmconv}) in the inequality. The lemma
is proved.
\end{proof}
Lemma \ref{lemmaone} and Lemma \ref{lemmatwo} imply the
theorem.
{\unskip\nobreak\hfill$\Box$\par\addvspace{\medskipamount}}
\subsection{Proofs of the results in Section \ref{timecon}}
For the proof of Theorem \ref{predictable} we will need the
following Lemma:
\begin{lemma}
\label{Q} In the setting of Theorem \ref{predictable}, let $\t$
be a stopping time such that $\s\leq \t<\t_\s.$ Let $H$ be a
bounded, $\F_{\t}$-measurable payoff. Then
$\Pi_{\s}(H)=\EQsigma{H}.$
\end{lemma}
\begin{proof}
As $\bar{\F}^Y_{\t_{i+1}-}=\bar{\F}^Y_{\t_{i}}$ for all $i$,
$H$ is $\bar{\F}^S_\t \vee \F_\s$-measurable. Consequently, the
lemma follows directly from the definition of
market-consistency.
\end{proof}

\textit{Proof of Theorem \ref{predictable}.} As the jump times
of $Y$ for the $i$-th jump, $(\t_i),$ are predictable
respectively, $\t_\s$ has to be predictable as well. This is
seen as follows: First of all note that by the predictability
of $\t_i$ there exists sequences of stopping times $(\t_i^n)$
with $\t^n_i<\t_i$ and $\t^n_i\uparrow \t_i$ as $n\to\infty.$
Define $\sigma^n:=\sum_{i=1}^\infty
I_{\{\t_i=\t_{\sigma}\}}(\t^n_i \vee \sigma).$ Let $\t_{0}:=0.$
Then $\{\t_i=\t_{\sigma}\}= \{\t_{i-1}\leq \sigma<\t_i\}\in
\F_{\sigma}$ for $i=1,2,\ldots$ (Since $\t_\s$ is the first
jump after time $\s$, we know at time $\s$, if we have observed
$i-1$ jumps so far, so that the next jump will be the $i$-th
one.) Therefore,
$$\{\sigma^n\leq t\}=\bigcup_{i=1}^\infty \Big( \{\t^n_i\leq t\}\cap \{\sigma\leq t\}\cap  \{\t_i=\t_{\sigma}\}\Big)\in \F_t  .$$
Thus, $\sigma^n$ is indeed a stopping time. Furthermore,
clearly $\sigma^n<\tau_\sigma$ and
$$\sigma^n\uparrow \sum_{i=1}^\infty
I_{\{\t_i=\t_{\sigma}\}}(\t_i \vee \sigma)=\sum_{i=1}^\infty
I_{\{\t_i=\t_{\sigma}\}}(\t_{\sigma} \vee \sigma)
=\sum_{i=1}^\infty I_{\{\t_i=\t_{\sigma}\}}\t_{\sigma}=
\tau_\sigma,$$ where we have used that $\s<\t_\s.$ Hence,
$\t_\s$ is indeed predictable.

Next, let $H\in L^\infty(\F_{\t_{\s}})$ and $A\in \F_{\s^m}$
for an $m\in\mathbb{N}.$ By time-consistency and the local
property of $(\Pi_\s)$ we get for all $n\geq m$
$$
\Pi_{\s}(H)= \Pi_{\s}( \Pi_{\s^n}(H I_A+HI_{A^c}))= \Pi_{\s}( I_A\Pi_{\s^n}(H) +I_{A^c}\Pi_{\s^n}(H)).
$$
Next observe that for $n\geq m,$ $I_A\Pi_{\s^n}(H)
+I_{A^c}\Pi_{\s^n}(H)$ is $\F_{\s^n}$-measurable. By Lemma
\ref{Q} this implies that
\begin{align*}
\Pi_{\s}( I_A\Pi_{\s^n}(H) +I_{A^c}\Pi_{\s^n}(H))&=
\EQsigma{I_A\Pi_{\s^n}(H) +I_{A^c}\Pi_{\s^n}(H)}\\
&=
\EQsigma{I_A\Pi_{\s^n}(H)} +\EQsigma{I_{A^c}\Pi_{\s^n}(H)}
\\
&=
\EQsigma{\Pi_{\s^n}(I_AH)} +\EQsigma{\Pi_{\s^n}(I_{A^c}H)}\\
&=
\Pi_{\s}( \Pi_{\s^n}(I_AH)) +\Pi_{\s}(\Pi_{\s^n}( I_{A^c}H))
=
\Pi_{\s}( I_AH) +\Pi_{\s}( I_{A^c}H),
\end{align*}
where we used Lemma \ref{lemmalocal} in the third and
time-consistency in the last equation. Hence, for every $A\in
\F_{\s^m}$ we have that \be \label{last3} \Pi_{\s}(H)=\Pi_{\s}(
I_AH) +\Pi_{\s}( I_{A^c}H). \ee Since by assumption $S$ is
continuous and $\bar{\F}^Y_{\t_{i+1}-}=\bar{\F}^Y_{\t_i}$ for
all $i$, $\bigcup_m \F_{\s^m}$ is a generating system for
$\F^S_{\t_\s}.$ By the continuity of $\Pi_\s$ this entails that
(\ref{last3}) holds for all $A\in \F^S_{\t_\s}.$ Thus, indeed
$\Pi_{\s}$ restricted to $\F_{\t_\s}$ satisfies the
market-local property. This proves the first part of the
theorem. The second part follows from Theorem \ref{main} and
Theorem \ref{main2}.
{\unskip\nobreak\hfill$\Box$\par\addvspace{\medskipamount}}

\subsection{Proofs of the results in Section \ref{g}} \textit{Proof of Theorem
\ref{g-expectation}.} First assume that
$g(t,z^f,z,\tilde{z})-\theta_tz^f$ does not depend on $z^f.$ We
will prove that for every $\t$, $\mathcal{E}^g_{\t}$ is
market-consistent, by using Proposition \ref{prop}
(i)$\Rightarrow$ (iii). Without loss of generality assume that
$\t=0.$ Let $H^S\in L^\infty(\bar{\F}^S_T).$ Denote by
$W^{*,f}$ the Brownian motion under $\mathbb{Q}$, i.e.,
$W^{*,f}_s=W^f_s-\int_0^s \theta_u du,$ where the integral is
defined componentwise. Since the financial market is complete
there exists a predictable $n$-dimensional process $Z^f\in
L^2(\bar{\F}^S ,d\mathbb{Q}\times du)$ such that $H^S=\int_0^T
Z^f_s dW^{*,f}_s.$ Define $Y_s=\EbarQSs{H^S}=\int_0^s Z^f_u
dW^{*,f}_u.$ Clearly, for every stopping time $\s$ we have that
$$\EbarStau{\int_\s^T |Z^f_u|^2 dW^{*,f}_u}\leq
2||Y||^2_{S^\infty}\leq 2||H||^2_\infty.$$ Predictable
processes satisfying such a boundedness property are also
called BMOs, see \citet{kazamaki1994continuous}. As $\theta$ is
bounded it follows from Theorem 3.24 in
\citet{barrieu:elkaroui:2009:carmona} that $Z^f$ is a BMO under
$\mathbb{P}$. In particular, $Z^f\in
L^2(\bar{\F}^S,d\mathbb{P}\times du).$ Now we get
\begin{align*}
dY_s
&=-g(s,0,0,0)ds+Z^f_s
dW^{*,f}_s\\
&=-(g(s,Z^f_s,0,0)-\theta_s Z^f_s)ds+Z^f_s
dW^{*,f}_s=-g(s,Z^f_s,0,0)ds+Z^f_s
dW^{f}_s,
\end{align*}
where we used that $g(s,0,0,0)=0$ in the first equation. In the
second equation we applied that
$g(s,0,0,0)=g(s,z^f,0,0)-\theta_sz^f$, as by assumption
$g(s,z^f,0,0)-\theta_sz^f$ does not depend on $z^f.$ In the
last equation we used the definition of $W^{*,f}$. This entails
that $Y_s=\EbarQSs{H}$ solves the BSDE with terminal condition
$H$ and driver $g.$ Therefore, indeed for every $H^S$ we have
that $\EQz{H^S}=\mathcal{E}^g_0(H^S),$ and it follows from
Proposition \ref{prop} (i)$\Rightarrow$(iii) (with
$\G=\{\Omega,\emptyset\}$) that $\mathcal{E}^g_0$ is
market-consistent. For general $\t$ the argument is similar.

Now let us prove the other direction. For arbitrary
$\bar{z}^f\in\mathbb{R}^n$ define
$$\bar{g}(t,z^f,z,\tilde{z}):=g(t,z^f+\bar{z}^f,z,\tilde{z})-\theta_t \bar{z}^f.$$
We need to show that $\bar{g}=g.$ Let
$(\mathcal{E}^{g}_t(H+\bar{z}^f W^{*,f}_T),Z^f,Z,\tilde{Z})$ be
the solution of the BSDE with terminal condition $H+\bar{z}^f
W^{*,f}_T$ and driver function $g.$ Note that the process
$Y^*_t:=\mathcal{E}^{g}_t(H+\bar{z}^f W^{*,f}_T)-\bar{z}^f
W^{*,f}_t$ is equal to $H$ at time $T.$ On the other hand we
have
\begin{align*}
dY^*_t&= -(g(t,Z^f_t,Z_t,\tilde{Z}_t)-\theta_t \bar{z}^f)dt+ (Z^f_t-\bar{z}^f)dW^f_t+ Z_t dW_t+\int_{\mathbb{R}\setminus\{0\}}\tilde{Z}_t(x) \tilde{N}(dt,dx)\\
&=-\bar{g}(t,Z^f_{new,t},Z_t,\tilde{Z}_t)dt+ Z^f_{new,t} dW^f_t+ Z_t
dW_t+\int_{\mathbb{R}\setminus\{0\}}\tilde{Z}_t(x) \tilde{N}(dt,dx),\end{align*} where
$Z^f_{new}:=Z^f-\bar{z}^f.$ Therefore, $Y^*$ solves the BSDE
with terminal condition $H$ and driver $\bar{g}.$ Hence, for
every $t$ we have $\mathcal{E}^{g}_t(H+\bar{z}^f
W^{*,f}_T)-\bar{z}^f
W^{*,f}_t=Y^*_t=\mathcal{E}^{\bar{g}}_t(H)$. Denote by
$\tilde{S}$ the vector of the discounted stock prices. By
market-consistency we can conclude that for every $H\in
L^\infty(\F_T)$
\begin{align}
\label{bar=}
\mathcal{E}^{g}_t(H)&=\mathcal{E}^{g}_t\Big(H+\int_0^T
\bar{z}^f \tilde{\s}^{-1}(t,e^{rt}\tilde{S}_t)d\tilde{S}_t
\Big)- \int_0^t
\bar{z}^f\tilde{\s}^{-1}(t,e^{rt}\tilde{S}_t)d\tilde{S}_t\nonumber\\&=\mathcal{E}^{g}_t(H+\bar{z}^f
W^{*,f}_T)-\bar{z}^f W^{*,f}_t
=\mathcal{E}^{\bar{g}}_t(H),
\end{align}
for all $t.$ Next choose $z^f\in\mathbb{R}^n,z\in\mathbb{R}^d$
and $\tilde{z}\in L^2(\nu(dx)).$ Set $$H:=-\int_0^T
g(s,z^f,z,\tilde{z})ds+z^f W^f_T+zW_T+\int_0^T
\int_{\mathbb{R}\setminus\{0\}}\tilde{z}(x)\tilde{N}(ds,dx).$$
Notice that $-\int_0^tg(s,z^f,z,\tilde{z})ds +\int_0^t z^f
dW^f_s+\int_0^tzdW_s+\int_0^t
\int_{\mathbb{R}\setminus\{0\}}\tilde{z}(x)\tilde{N}(ds,dx)$ by
definition is the solution of the BSDE with terminal condition
$H$ and driver $g.$ In particular, $\mathcal{E}^g_0(H)=0.$ By
(\ref{bar=}) this yields
$\mathcal{E}^{\bar{g}}_0(H)=\mathcal{E}^g_0(H)=0$ and
\begin{align*}
-\int_0^t &g(s,z^f,z,\tilde{z})ds+\int_0^t z^f
dW^f_s+\int_0^tzdW_s+\int_0^t
\int_{\mathbb{R}\setminus\{0\}}\tilde{z}(x)\tilde{N}(ds,dx)\\
&=\mathcal{E}^g_t(H)=\mathcal{E}^{\bar{g}}_t(H)\\
&=\mathcal{E}^{\bar{g}}_0(H)
-\int_0^t \bar{g}(s,Z^f_s,Z_s,\tilde{Z}_s)ds+\int_0^t Z^f_s
dW^f_s+\int_0^tZ_s dW_s+\int_0^t
\int_{\mathbb{R}\setminus\{0\}}\tilde{Z}_s(x)\tilde{N}(ds,dx)
\\
&=
-\int_0^t \bar{g}(s,Z^f_s,Z_s,\tilde{Z}_s)ds+\int_0^t Z^f_s
dW^f_s+\int_0^tZ_s dW_s+\int_0^t
\int_{\mathbb{R}\setminus\{0\}}\tilde{Z}_s(x)\tilde{N}(ds,dx),
\end{align*}
where $(Z^f,Z,\tilde{Z})$ belong to the solution of the
$\bar{g}$-expectation with terminal condition $H.$ By the
uniqueness of the decomposition of semi-martingales this
entails that \begin{align} \label{one}\int_0^t z^f
dW^f_s&+\int_0^tzdW_s+\int_0^t
\int_{\mathbb{R}\setminus\{0\}}\tilde{z}(x)\tilde{N}(ds,dx)\nonumber\\
&=\int_0^t Z^f_s dW^f_s+\int_0^tZ_s dW_s+\int_0^t
\int_{\mathbb{R}\setminus\{0\}}\tilde{Z}_s(x)\tilde{N}(ds,dx),\end{align}
and \be \label{two} \int_0^t g(s,z^f,z,\tilde{z})ds= \int_0^t
\bar{g}(s,Z^f_s,Z_s,\tilde{Z}_s)ds.\ee Taking for instance the
quadratic covariation with respect to the components of $W^f$,
$W$, and with respect to $\tilde{N}$ in (\ref{one})
respectively, we may conclude that $Z^f_t=z^f,$ $Z_t=z,$
$d\mathbb{P}\times dt$ a.s., and $\tilde{Z}_t=\tilde{z},$
$\nu(dx)\times d\mathbb{P}\times dt$. But then (\ref{two})
yields that for a.s. all $\omega$ \beas \int_0^t
g(s,z^f,z,\tilde{z})ds= \int_0^t
\bar{g}(s,z^f,z,\tilde{z})ds,\mbox{ for all }t\in[0,T]\eeas and
therefore $g(t,z^f,z,\tilde{z})=\bar{g}(t,z^f,z,\tilde{z})$ for
a.s. all $\omega$ for Lesbegue a.s. all $t.$
{\unskip\nobreak\hfill$\Box$\par\addvspace{\medskipamount}}

\subsection{Proofs of the results in Section \ref{g2}}
{\it Proof of Proposition \ref{mean-var}.} First of all note
that since $\frac{d\mathbb{Q}^h}{d\mathbb{P}}$ is
$\sigma(W^{f,*}_t|0\leq t\leq T)$-measurable, we have that $W$
and $\tilde{N}$ have the same joint distribution under
$\mathbb{Q}^h$ as under $\mathbb{P}$ (since they are
independent of $W^{f,*}$).

By well known projection results, there exists adapted
$Z^{h,f}_{ih}:\Omega\to \mathbb{R}^n,$ $Z^h_{ih}:\Omega\to
\mathbb{R}^d,$ measurable with respect to $\F_{ih},$
$\tilde{Z}^h_{ih}:\Omega\times \mathbb{R}\setminus\{0\}\to
\mathbb{R},$ measurable with respect to
$\F_{ih}\otimes\mathcal{B}(\mathbb{R}\setminus\{0\}),$ and a
real-valued $\mathbb{Q}^h$-martingale $(L^h_{ih})_i$ which is
orthogonal (under $\mathbb{Q}^h$) to $W^{f,*}_{ih},$ $W_{ih},$
and $ \tilde{N}((0,ih],dx),$ such that
\begin{align*}
\Pi_{(i+1)h}
(H)&=\EQh{\Pi_{(i+1)h}(H)}+ Z^{h,f}_{ih} \Delta W^{f,*}_{(i+1)h}+ Z^h_{ih} \Delta W_{(i+1)h}\\
&\hspace{0.5cm}+\int_{\mathbb{R}\setminus\{0\}}
\tilde{Z}^h_{ih}(x) \tilde{N}((ih,(i+1)h],dx)+\Delta L^h_{(i+1)h}.
\end{align*}
For the sake of simplicity we will omit the superscript $h$ for
the $Z^{h,f},Z^h$, and $\tilde{Z}^h$ in the sequel.

It follows that
\begin{align}
\label{delta} \Delta& \Pi_{(i+1)h}(H)\nonumber\\
&= \Pi_{(i+1)h}
(H)- \Pi_{ih}(H)\nonumber\\
&=
 \Pi_{(i+1)h}
(H)- \Pi^v_{ih,(i+1)h}(\Pi_{(i+1)h}(H))
\nonumber\\
&= Z^f_{ih} \Delta W^{f,*}_{(i+1)h}+ Z_{ih} \Delta W_{(i+1)h}+\int_{\mathbb{R}\setminus\{0\}}
\tilde{Z}_{ih}(x) \tilde{N}((ih,(i+1)h],dx)+\Delta L_{(i+1)h}\nonumber
\\
&\hspace{0.3cm}-\Pi^v_{ih}
\bigg( Z^f_{jh} \Delta W^{f,*}_{(i+1)h}+ Z_{ih} \Delta W_{(i+1)h}+\int_{\mathbb{R}\setminus\{0\}}
\tilde{Z}_{ih}(x) \tilde{N}((ih,(i+1)h],dx)+\Delta L_{(i+1)h}\bigg)
\nonumber\\
&= Z^f_{jh} \Delta W^{f,*}_{(i+1)h}+Z_{ih} \Delta W_{(i+1)h}+\int_{\mathbb{R}\setminus\{0\}}
\tilde{Z}_{ih}(x) \tilde{N}((ih,(i+1)h],dx)+\Delta L_{(i+1)h}
\nonumber\\
&\hspace{0.3cm}-\Pi^v_{ih,(i+1)h}
\bigg( Z_{ih} \Delta W_{(i+1)h}+\int_{\mathbb{R}\setminus\{0\}}
\tilde{Z}_{ih}(x) \tilde{N}((ih,(i+1)h],dx)+\Delta L_{(i+1)h}\bigg)
\nonumber\\
&\stackrel{!}{=}Z^f_{jh} \Delta W^{f}_{(i+1)h}+Z_{ih} \Delta W_{(i+1)h}+\int_{\mathbb{R}\setminus\{0\}}
\tilde{Z}_{ih}(x) \tilde{N}((ih,(i+1)h],dx)+\Delta L_{(i+1)h}
\nonumber\\
&\hspace{0.3cm}-\Big[\theta_{ih}Z^f_{ih}+\frac{\a}{2}|Z_{ih}|^2 +\frac{\a}{2}
\int_{\mathbb{R}\setminus\{0\}}|\tilde{Z}_{ih}(x)|^2\nu(dx)
\Big]h\nonumber\\
&\hspace{0.3cm}-\frac{\a}{2}\EQti{(\Delta L_{(i+1)h}-\Enip{\Delta
L_{(i+1)h}})^2},
\end{align}
where we have used (\ref{rec}) in the third equation.
Furthermore, we applied cash invariance in the third and
market-consistency in the fourth equation. To see that the last
equation holds denote by
$Cov_{\mathbb{P}_{\F^{S^h}_{(i+1)h}}}(X_1,X_2)$ the covariance
of $X_1$ and $X_2$ with respect to
$\mathbb{P}_{\F^{S^h}_{(i+1)h}}.$ Since all random variables
are $\F_{(i+1)h}$-measurable we may assume that
$\frac{d\mathbb{Q}^h}{d\mathbb{P}}$ is
$\F^{S^h}_{(i+1)h}$-measurable.

It is
\begin{align}
\label{last}
\Pi^v_{ih,(i+1)h}\bigg( Z_{ih} &\Delta W_{(i+1)h}+\int_{\mathbb{R}\setminus\{0\}}
\tilde{Z}_{ih}(x) \tilde{N}((ih,(i+1)h],dx)+\Delta L_{(i+1)h}\bigg)\nonumber\\
&=\frac{\a}{2}\EQti{Var_{\F^{S^h}_{(i+1)h}}\bigg( Z_{ih} \Delta W_{(i+1)h}+\int_{\mathbb{R}\setminus\{0\}}
\tilde{Z}_{ih}(x) \tilde{N}((ih,(i+1)h],dx)+\Delta L_{(i+1)h}\bigg)}\nonumber\\
&=\frac{\a}{2}\bigg(h|Z_{ih}|^2 +h\int_{\mathbb{R}\setminus\{0\}}
|\tilde{Z}_{ih}(x)|^2 \nu(dx)+\EQti{\big(\Delta
L_{(i+1)h}-\Enip{\Delta
L_{(i+1)h}}\big)^2}
\nonumber\\&\hspace{0.5cm}+ 2\sum_{j=1}^d Z^j_{ih}\EQti{
Cov_{\mathbb{P}_{\F^{S^h}_{(i+1)h}}}( \Delta W^j_{(i+1)h},\Delta L_{(i+1)h})}\nonumber\\
&\hspace{0.5cm}
+2\EQti{Cov_{\mathbb{P}_{\F^{S^h}_{(i+1)h}}}(\int_{\mathbb{R}\setminus\{0\}}\tilde{Z}_{ih}(x)
\tilde{N}((ih,(i+1)h],dx),\Delta L_{(i+1)h})}
\nonumber\\&\hspace{0.5cm}+ 2\sum_{j=1}^dZ^j_{ih}\EQti{
Cov_{\mathbb{P}_{\F^{S^h}_{(i+1)h}}}( \Delta W^j_{(i+1)h},\int_{\mathbb{R}\setminus\{0\}}\tilde{Z}_{ih}(x)
\tilde{N}((ih,(i+1)h],dx))}\bigg),
\end{align}
where we have used that $W$ and $\tilde{N}$ are independent of
$S^h$ and that $\frac{d\mathbb{Q}^h}{d\mathbb{P}}$ is
$\F^{S^h}_{(i+1)h}$-measurable.

Hence, to prove (\ref{delta}) the only thing what remains left
is to show that the the covariance terms in (\ref{last}) are
zero. By Lemma \ref{A.12} (with $\hat{\mathbb{P}}=\mathbb{Q}^h$
and $\bar{\mathbb{P}}=\mathbb{P}$) we get for $j=1,\ldots,d$
$$\ESip{\Delta
W^j_{(i+1)h}}=\Enip{\Delta
W^j_{(i+1)h}}=\Ei{\Delta
W^j_{(i+1)h}}=0,
$$
where we used the independence of $S^h$ and $W$ in the second
equation. Hence, since $L_{(i+1)h}$ and $\Delta W_{(i+1)h}$ are
orthogonal under $\mathbb{Q}^h_{\F_{ih}},$ we may conclude that
 \begin{align*}
 \EQti{
Cov_{\mathbb{P}_{\F^{S^h}_{(i+1)h}}}\Big( \Delta W^j_{(i+1)h},\Delta
L_{(i+1)h}\Big)}&=\EQti{
 \Delta W^j_{(i+1)h}\Delta
L_{(i+1)h}}\\
&=Cov_{\mathbb{Q}^h_{\F_{ih}}}( \Delta
W^j_{(i+1)h},\Delta L_{(i+1)h})=0.
\end{align*} Similarly, it may be seen
that the second and third covariance terms in (\ref{delta}) are
zero. Thus, (\ref{delta}) is proved. From (\ref{delta}) we may
finally conclude
\begin{align*}
\Pi_T&(H)-\Pi_{ih}(H)\\
&=\sum_{j=i}^{T/h-1}
\Delta\Pi_{(j+1)h}(H)\nonumber\\
&=\sum_{j=i}^{T/h-1}\bigg(Z^f_{jh} \Delta W^{f}_{(j+1)h}+Z_{jh} \Delta W_{(j+1)h}+\int_{\mathbb{R}\setminus\{0\}}
\tilde{Z}_{jh}(x) \tilde{N}((jh,(j+1)h],dx)+\Delta L_{(j+1)h}\bigg)
\nonumber\\
&\hspace{0.3cm}-\sum_{j=i}^{T/h-1}\Big[\theta_{jh}Z^f_{jh}
+\frac{\a}{2}|Z_{jh}|^2 +\frac{\a}{2}\int_{\mathbb{R}\setminus\{0\}}|\tilde{Z}_{jh}(x)|^2\nu(dx)\Big]h
\\
&\hspace{0.3cm}+\frac{\a}{2}\EQj{\big(\Delta L_{(j+1)h}-\ESjp{\Delta
L_{(j+1)h}}\big)^2}.
\end{align*}
Since by construction $\Pi_T(H)=H$ the proposition is proved.
{\unskip\nobreak\hfill$\Box$\par\addvspace{\medskipamount}}

\bibliographystyle{apalike}
\bibliography{APbib}

\end{document}